\documentclass[11pt]{article}
\usepackage{amsmath,amssymb,amsfonts,amsthm}
\usepackage{subcaption}
\usepackage{graphicx}
\usepackage{fullpage}
\usepackage[backref=page]{hyperref}
\usepackage{color}
\usepackage{wrapfig}
\usepackage{tikz}
\usetikzlibrary{decorations.pathreplacing}
\usepackage{setspace}
\usepackage{algorithm}
\usepackage[noend]{algpseudocode}
\usepackage[framemethod=tikz]{mdframed}
\usepackage{xspace}
\usepackage{pgfplots}
\usepackage{framed}
\usepackage{thmtools}
\usepackage{thm-restate}
\usepackage{booktabs}
\pgfplotsset{compat=1.5}

\newtheorem{theorem}{Theorem}[section]

\newtheorem{lemma}[theorem]{Lemma}

\newtheorem{definition}[theorem]{Definition}

\newenvironment{proofof}[1]{\begin{trivlist} \item {\bf Proof
#1:~~}}
  {\qed\end{trivlist}}

\newcommand{\namedref}[2]{\hyperref[#2]{#1~\ref*{#2}}}
\newcommand{\thmlab}[1]{\label{thm:#1}}
\newcommand{\thmref}[1]{\namedref{Theorem}{thm:#1}}
\newcommand{\lemlab}[1]{\label{lem:#1}}
\newcommand{\lemref}[1]{\namedref{Lemma}{lem:#1}}

\newcommand{\seclab}[1]{\label{sec:#1}}
\newcommand{\secref}[1]{\namedref{Section}{sec:#1}}

\newcommand{\figlab}[1]{\label{fig:#1}}
\newcommand{\figref}[1]{\namedref{Figure}{fig:#1}}
\newcommand{\alglab}[1]{\label{alg:#1}}
\renewcommand{\algref}[1]{\namedref{Algorithm}{alg:#1}}
\newcommand{\tablelab}[1]{\label{tab:#1}}
\newcommand{\tableref}[1]{\namedref{Table}{tab:#1}}

\newcommand{\eqnlab}[1]{\label{eq:#1}}
\newcommand{\eqnref}[1]{\namedref{Equation}{eq:#1}}

\def \Cost    {\mdef{\mathsf{Cost}}}

\newcommand{\PPr}[1]{\ensuremath{\mathbf{Pr}\left[#1\right]}}

\newcommand{\Ex}[1]{\ensuremath{\mathbb{E}\left[#1\right]}}

\renewcommand{\O}[1]{\ensuremath{\mathcal{O}\left(#1\right)}}
\newcommand{\tO}[1]{\ensuremath{\tilde{\mathcal{O}}\left(#1\right)}}
\newcommand{\eps}{\varepsilon}

\def \ba    {\mdef{\mathbf{a}}}

\def \bA    {\mdef{\mathbf{A}}}
\def \bB    {\mdef{\mathbf{B}}}
\def \bD    {\mdef{\mathbf{D}}}
\def \bG    {\mdef{\mathbf{G}}}

\def \bM    {\mdef{\mathbf{M}}}

\def \bS    {\mdef{\mathbf{S}}}
\def \bT    {\mdef{\mathbf{T}}}

\def \bV    {\mdef{\mathbf{V}}}
\def \bU    {\mdef{\mathbf{U}}}
\def \bX    {\mdef{\mathbf{X}}}

\def \bx    {\mdef{\mathbf{x}}}

\def \bz    {\mdef{\mathbf{z}}}

\newcommand{\mdef}[1]{{\ensuremath{#1}}\xspace}  
\DeclareMathOperator*{\argmin}{argmin}

\DeclareMathOperator*{\poly}{poly}
\DeclareMathOperator*{\Var}{Var}

\def \etal{{\it et~al.}}

\newcommand{\flr}[1]{\mdef{\left\lfloor#1\right\rfloor}}              

\newcommand{\ignore}[1]{}

\newif\ifnotes\notestrue 
\ifnotes
\newcommand{\samson}[1]{\textcolor{purple}{{\bf (Samson:} {#1}{\bf ) }} \marginpar{\tiny\bf
             \begin{minipage}[t]{0.5in}
               \raggedright S:
            \end{minipage}}}            							
\else
\newcommand{\samson}[1]{}
\fi

\makeatletter
\renewcommand*{\@fnsymbol}[1]{\textcolor{mahogany}{\ensuremath{\ifcase#1\or *\or \dagger\or \ddagger\or
 \mathsection\or \triangledown\or \mathparagraph\or \|\or **\or \dagger\dagger
   \or \ddagger\ddagger \else\@ctrerr\fi}}}
\makeatother

\providecommand{\email}[1]{\href{mailto:#1}{\nolinkurl{#1}\xspace}}

\definecolor{bleudefrance}{rgb}{0.19, 0.55, 0.91}
\definecolor{mahogany}{rgb}{0.75, 0.25, 0.0}

\hypersetup{
     colorlinks   = true,
     citecolor    = mahogany,
     urlcolor     = mahogany,
	 linkcolor	  = bleudefrance
}

\usepackage{wrapfig}
\usepackage{mathtools}
\newcommand{\erclogowrapped}[1]{%
\setlength\intextsep{0pt}%
\begin{wrapfigure}[3]{r}{#1*\real{1.1}}%
\includegraphics[width=#1]{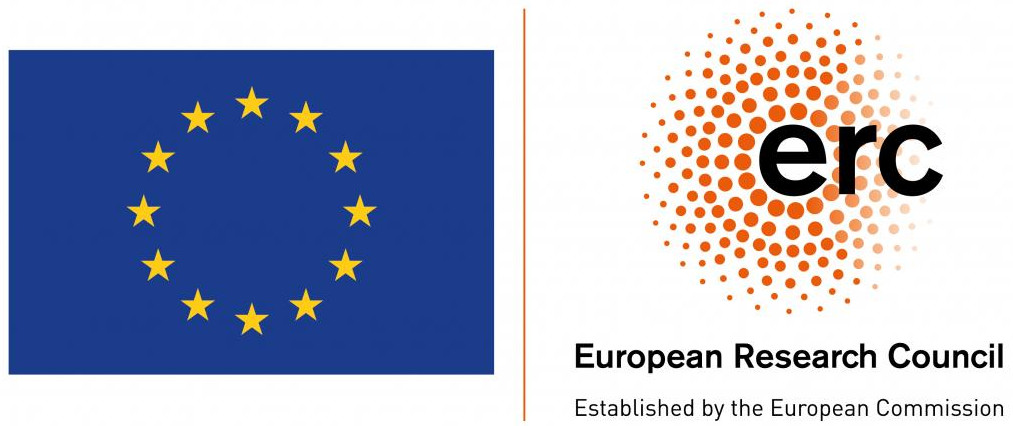}%
\end{wrapfigure}%
}

\makeatletter
\def\blfootnote{\xdef\@thefnmark{}\@footnotetext}
\makeatother

\begin{document}

\allowdisplaybreaks

\title{Fast $(1+\eps)$-Approximation Algorithms for Binary Matrix Factorization}
\author{Ameya Velingker\thanks{Google Research. E-mail: \email{ameyav@google.com}.}
\and
Maximilian V\"{o}tsch\thanks{Faculty of Computer Science, Univie Doctoral School Computer Science DoCS, University of Vienna. Email: \email{maximilian.voetsch@univie.ac.at}.}
\and
David P. Woodruff\thanks{Carnegie Mellon University. E-mail: \email{dwoodruf@andrew.cmu.edu}. Work done while at Google Research.}
\and
Samson Zhou\thanks{UC Berkeley and Rice University. E-mail: \email{samsonzhou@gmail.com}.}
}
\date{\today}

\maketitle{}

\begin{abstract}
We introduce efficient $(1+\eps)$-approximation algorithms for the binary matrix factorization (BMF) problem, where the inputs are a matrix $\mathbf{A}\in\{0,1\}^{n\times d}$, a rank parameter $k>0$, as well as an accuracy parameter $\varepsilon>0$, and the goal is to approximate $\mathbf{A}$ as a product of low-rank factors $\mathbf{U}\in\{0,1\}^{n\times k}$ and $\mathbf{V}\in\{0,1\}^{k\times d}$. Equivalently, we want to find $\bU$ and $\bV$ that minimize the Frobenius loss $\|\mathbf{U}\mathbf{V} - \mathbf{A}\|_F^2$. Before this work, the state-of-the-art for this problem was the approximation algorithm of Kumar \etal~[ICML 2019], which achieves a $C$-approximation for some constant $C\ge 576$. We give the first $(1+\eps)$-approximation algorithm using running time singly exponential in $k$, where $k$ is typically a small integer. Our techniques generalize to other common variants of the BMF problem, admitting bicriteria $(1+\eps)$-approximation algorithms for $L_p$ loss functions and the setting where matrix operations are performed in $\mathbb{F}_2$. Our approach can be implemented in standard big data models, such as the streaming or distributed models. 
\blfootnote{\erclogowrapped{5\baselineskip} \emph{M. Vötsch:} This project has received funding from the European Research Council (ERC) under the European Union's Horizon 2020 research and innovation programme (Grant agreement No. 101019564 ``The Design of Modern Fully Dynamic Data Structures (MoDynStruct)''.}
\end{abstract}

\section{Introduction}
Low-rank approximation is a fundamental tool for factor analysis. The goal is to decompose several observed variables stored in the matrix $\bA \in \mathbb{R}^{n \times d}$ into a combination of $k$ unobserved and uncorrelated variables called factors, represented by the matrices $\bU \in \mathbb{R}^{n \times k}$ and $\bV \in \mathbb{R}^{k \times d}$. In particular, we want to solve the problem 
\[\min_{\bU\in\mathbb{R}^{n\times k}, \bV\in\mathbb{R}^{k\times d}} \|\bU \bV\ - \bA\|,\]
for some predetermined norm $\|\cdot\|$. 
Identifying the factors can often decrease the number of relevant features in an observation and thus significantly improve interpretability. Another benefit is that low-rank matrices allow us to approximate the matrix $\bA$ with its factors $\bU$ and $\bV$ using only $(n+d)k$ parameters rather than the $nd$ parameters needed to represent $\bA$. Moreover, for a vector $\bx\in\mathbb{R}^d$, we can approximate the matrix-vector multiplication $\bA \bx \approx \bU\bV \bx$ in time $(n+d)k$, while computing $\bA\bx$ requires $nd$ time. 
These benefits make low-rank approximation one of the most widely used tools in machine learning, recommender systems, data science, statistics, computer vision, and natural language processing.
In many of these applications, discrete or categorical datasets are typical. In this case, restricting the underlying factors to a discrete domain for interpretability often makes sense.
For example, \cite{KumarPRW19} observed that nearly half of the data sets in the UCI repository~\cite{UCI17} are categorical and thus can be represented as binary matrices, possibly using multiple binary variables to represent each category. 

In the binary matrix factorization (BMF) problem, the input matrix $\bA\in\{0,1\}^{n\times d}$ is binary. Additionally, we are given an integer range parameter $k$, with $0 < k \ll n, d$.
The goal is to approximate $\bA$ by the factors $\bU\in\{0,1\}^{n\times k}$ and $\bV\in\{0,1\}^{k\times d}$ such that $\bA\approx\bU\bV$.
The BMF problem restricts the general low-rank approximation problem to a discrete space, making finding good factors more challenging (see \secref{sec:related}).

\subsection{Our Contributions}
We present $(1+\eps)$-approximation algorithms for the binary low-rank matrix factorization problem for several standard loss functions used in the general low-rank approximation problem.
\tableref{table:summary} summarizes our results. 

\begin{table}[!htb]
\centering
\begin{tabular}{|c|c|c|c|}\hline
Reference & Approximation & Runtime & Other\\\hline
\cite{KumarPRW19} & $C\ge 576$ & $2^{\tO{k^2}}\poly(n,d)$ & Frobenius loss \\\hline
\cite{FominGLP020} & $1+\eps$ & $2^{\frac{2^{\O{k}}}{\eps^2}\log^2\frac{1}{\eps}}\poly(n,d)$ & Frobenius loss \\\hline
Our work & $1+\eps$ & $2^{\tO{k^2/\eps^4}}\poly(n,d)$ & Frobenius loss \\\hline\hline
\cite{KumarPRW19} & $C\ge 122^{2p-2}+2^{p-1}$ & $2^{\poly(k)}\poly(n,d)$ & $L_p$ loss, $p\ge 1$ \\\hline
Here & $1+\eps$ & $2^{\poly(k/\eps)}\poly(n,d)$ & $L_p$ loss, $p\ge 1$, bicriteria \\\hline\hline
\cite{FominGLP020} & $1+\eps$ & $2^{\frac{2^{\O{k}}}{\eps^2}\log^2\frac{1}{\eps}}\poly(n,d)$ & Binary field \\\hline
\cite{BanBBKLW19} & $1+\eps$ & $2^{\frac{2^{\O{k}}}{\eps^2}\log\frac{1}{\eps}}\poly(n,d)$ & Binary field \\\hline
\cite{KumarPRW19} & $C\ge40001$ & $2^{\poly(k)}\poly(n,d)$ & Binary field, bicriteria \\\hline
Our work & $1+\eps$ & $2^{\poly(k/\eps)}\poly(n,d)$ & Binary field, bicriteria\\\hline
\end{tabular}
\caption{Summary of related work on binary matrix factorization}
\tablelab{table:summary}
\end{table}

\paragraph{Binary matrix factorization.}
We first consider the minimization of the Frobenius norm, defined by $\|\bA-\bU\bV\|_F^2=\sum_{i\in[n]}\sum_{j\in d}|\bA_{i,j}-(\bU\bV)_{i,j}|^2$, where $[n]:=\{1,\ldots,n\}$ and $\bA_{i,j}$ denotes the entry in the $i$-th row and the $j$-th column of $\bA$. Intuitively, we can view this as finding a least-squares approximation of $\bA$.

We introduce the first $(1+\eps)$-approximation algorithm for the BMF problem that runs in singly exponential time. 
That is, we present an algorithm that, for any $\eps>0$, returns $\bU'\in\{0,1\}^{n\times k},\bV'\in\{0,1\}^{k\times d}$ with  
\[\|\bA-\bU'\bV'\|_F^2\le(1+\eps)\min_{\bU\in\{0,1\}^{n\times k},\bV\in\{0,1\}^{k\times d}}\|\bA-\bU\bV\|_F^2.\]
For $\eps\in(0,1)$, our algorithm uses $2^{\tO{k^2/\eps^4}}\poly(n,d)$ runtime and for $\eps\ge 1$, our algorithm uses $2^{\tO{k^2}}\poly(n,d)$ runtime, where $\poly(n,d)$ denotes a polynomial in $n$ and $d$.  

By comparison, \cite{KumarPRW19} gave a $C$-approximation algorithm for the BMF problem also using runtime $2^{\tO{k^2}}\poly(n,d)$, but for some constant $C\ge 576$. 
Though they did not attempt to optimize for $C$, their proofs employ multiple triangle inequalities that present a constant lower bound of at least $2$ on $C$.
See \secref{sec:overview} for a more thorough discussion of the limitations of their approach. 
\cite{FominGLP020} introduced a $(1+\eps)$-approximation algorithm for the BMF problem with rank-$k$ factors. However, their algorithm uses time doubly exponential in $k$, specifically $2^{\frac{2^{\O{k}}}{\eps^2}\log^2\frac{1}{\eps}}\poly(n,d)$, which \cite{BanBBKLW19} later improved to doubly exponential runtime $2^{\frac{2^{\O{k}}}{\eps^2}\log\frac{1}{\eps}}\poly(n,d)$, while also showing that time $2^{k^{\Omega(1)}}$ is necessary even for constant-factor approximation, under the Small Set Expansion Hypothesis and the Exponential Time Hypothesis.

\paragraph{BMF with $L_p$ loss.}
We also consider the more general problem of minimizing for $L_p$ loss for a given $p$, defined as the optimization problem of minimizing $\|\bA-\bU\bV\|_p^p=\sum_{i\in[n]}\sum_{j\in d}|\bA_{i,j}-(\bU\bV)_{i,j}|^p$. 
Of particular interest is the case $p=1$, which corresponds to robust principal component analysis, and which has been proposed as an alternative to Frobenius norm low-rank approximation that is more robust to outliers, i.e., values that are far away from the majority of the data points~\cite{ke2003robust,KeK05,kwak2008principal,ZhengLSYO12,brooks2013pure,MarkopoulosKP14,SongWZ17,park2018three,BanBBKLW19,MahankaliW21}. 
On the other hand, for $p>2$, low-rank approximation with $L_p$ error increasingly places higher priority on outliers, i.e., the larger entries of $\bU\bV$. 

We present the first $(1+\eps)$-approximation algorithm for the BMF problem that runs in singly exponential time, albeit at the cost of incurring logarithmic increases in the rank $k$, making it a bicriteria algorithm. 
Specifically, for any $\eps>0$, our algorithm returns $\bU'\in\{0,1\}^{n\times k'},\bV'\in\{0,1\}^{k'\times d}$ with  
\[\|\bA-\bU'\bV'\|_p^p\le(1+\eps)\min_{\bU\in\{0,1\}^{n\times k},\bV\in\{0,1\}^{k\times d}}\|\bA-\bU\bV\|_p^p,\]
where $k'=\O{\frac{k\log^2 n}{\eps^2}}$. 
For $\eps\in(0,1)$, our algorithm uses $2^{\poly(k/\eps)}\poly(n,d)$ runtime and for $\eps\ge 1$, our algorithm uses $2^{\poly(k)}\poly(n,d)$ runtime. 

Previous work~\cite{KumarPRW19} gave a $C$-approxmiation algorithm for this problem, using singly exponential runtime $2^{\poly(k)}\poly(n,d)$, without incurring a bicriteria loss in the rank $k$. However, their constant $C \ge 122^{2p-2} + 2^{p-1}$ is large and depends on $p$.
Again, their use of multiple triangle inequalities in their argument bars this approach from being able to achieve a $(1+\eps)$-approximation. 
To our knowledge, no prior works achieved $(1+\eps)$-approximation to BMF with $L_p$ loss in singly exponential time.

\paragraph{BMF on binary fields.}
Finally, we consider the case where all arithmetic operations are performed modulo two, i.e., in the finite field $\mathbb{F}_2$. 
Specifically, the $(i,j)$-th entry of $\bU\bV$ is the inner product $\langle\bU_i,\bV^{(j)}\rangle$ of the $i$-th row of $\bU$ and the $j$-th column of $\bV$, taken over $\mathbb{F}_2$. 
This model has been frequently used for dimensionality reduction for high-dimensional data with binary attributes~\cite{KoyuturkG03,ShenJY09,jiang2014clustering,DanHJ0Z18} and independent component analysis, especially in the context of signal processing~\cite{Yeredor11,GutchGYT12,painsky2015generalized,painsky2018linear}. 
This problem is also known as bipartite clique cover, the discrete basis problem, or minimal noise role mining and has been well-studied in applications to association rule mining, database tiling, and topic modeling~\cite{SeppanenBM03,SingliarH06,vaidya2007role,miettinen2008discrete,BelohlavekV10,LuVAH12,ChandranIK16,Chen0T022}.  

We introduce the first bicriteria $(1+\eps)$-approximation algorithm for the BMF problem on binary fields that runs in singly exponential time. 
Specifically, for any $\eps>0$, our algorithm returns $\bU'\in\{0,1\}^{n\times k'},\bV'\in\{0,1\}^{k'\times d}$ with  
\[\|\bA-\bU'\bV'\|_p^p\le(1+\eps)\min_{\bU\in\{0,1\}^{n\times k},\bV\in\{0,1\}^{k\times d}}\|\bA-\bU\bV\|_p^p,\]
where $k'=\O{\frac{k\log n}{\eps}}$ and all arithmetic operations are performed in $\mathbb{F}_2$. 
For $\eps\in(0,1)$, our algorithm has running time $2^{\poly(k/\eps)}\poly(n,d)$ and for $\eps\ge 1$, our algorithm has running time $2^{\poly(k)}\poly(n,d)$. 

By comparison, \cite{KumarPRW19} gave a bicriteria $C$-approximation algorithm for the BMF problem on binary fields with running time $2^{\poly(k)}\poly(n,d)$, for some constant $C\ge 40001$. 
Even though their algorithm also gives a bicriteria guarantee, their approach, once again, inherently cannot achieve $(1+\eps)$-approximation. 
On the other hand, \cite{FominGLP020} achieved a $(1+\eps)$-approximation without a bicriteria guarantee, but their algorithm uses doubly exponential running time $2^{\frac{2^{\O{k}}}{\eps^2}\log^2\frac{1}{\eps}}\poly(n,d)$, which \cite{BanBBKLW19} later improved to doubly exponential running time $2^{\frac{2^{\O{k}}}{\eps^2}\log\frac{1}{\eps}}\poly(n,d)$, while also showing that running time doubly exponential in $k$ is necessary for $(1+\eps)$-approximation on $\mathbb{F}_2$.

\paragraph{Applications to big data models.} 
We remark that our algorithms are conducive to big data models. 
Specifically, our algorithmic ideas facilitate a two-pass algorithm in the streaming model, where either the rows or the columns of the input matrix arrive sequentially, and the goal is to perform binary low-rank approximation while using space sublinear in the size of the input matrix.
Similarly, our approach can be used to achieve a two-round protocol in the distributed model, where either the rows or the columns of the input matrix are partitioned among several players, and the goal is to perform binary low-rank approximation while using total communication sublinear in the size of the input matrix. 
See \secref{sec:big:data} for a formal description of the problem settings and additional details. 

\subsection{Overview of Our Techniques}
\seclab{sec:overview}
This section briefly overviews our approaches to achieving $(1+\eps)$-approximation to the BMF problem. 
Alongside our techniques, we discuss why prior approaches for BMF fail to achieve $(1+\eps)$-approximation.

The BMF problem under the Frobenius norm is stated as follows:
Let $\bU^*\in\{0,1\}^{n\times k}$ and $\bV^*\in\{0,1\}^{k\times d}$ be optimal low-rank factors, so that
\begin{equation}\eqnlab{bmf}
\|\bU^*\bV^*-\bA\|_F^2=\min_{\bU\in\{0,1\}^{n\times k},\bV\in\{0,1\}^{k\times d}}\|\bU\bV-\bA\|_F^2.
\end{equation}
Our approach relies on the sketch-and-solve paradigm, and we ask of our sketch matrix $\bS$ that it is an \emph{affine embedding}, that is, given $\bU^*$ and $\bA$, for all $\bV\in\{0,1\}^{k\times d}$,
\[(1-\eps)\|\bU^*\bV-\bA\|_F^2\le\|\bS\bU^*\bV-\bS\bA\|_F^2\le(1+\eps)\|\bU^*\bV-\bA\|_F^2.\]
Observe that if $\bS$ is an affine embedding, then we obtain a $(1+\eps)$-approximation by solving for the minimizer $\bV^*$ in the sketched space. 
That is, given $\bS$ and $\bU^*$, instead of solving \eqnref{bmf} for $\bV^*$, it suffices to solve 
\[\argmin_{\bV\in\{0,1\}^{k\times d}}\|\bS\bU^*\bV-\bS\bA\|_F^2.\]


\paragraph{Guessing the sketch matrix $\bS$.}
A general approach taken by \cite{RazenshteynSW16,KumarPRW19,BanWZ19} for various low-rank approximation problems is first to choose $\bS$ in a way so that there are not too many possibilities for the matrices $\bS\bU^*$ and $\bS\bA$ and then find the minimizer $\bV^*$ for all guesses of $\bS\bU^*$ and $\bS\bA$. 
Note that this approach is delicate because it depends on the choice of the sketch matrix $\bS$. 
For example, if we chose $\bS$ to be a dense matrix with random Gaussian entries, then since there are $2^{nk}$ possibilities for the matrix $\bU^*\in\{0,1\}^{n\times k}$, we cannot enumerate the possible matrices $\bS\bU^*$. 
Prior work~\cite{RazenshteynSW16,KumarPRW19,BanWZ19} made the key observation that if $\bA$ (and thus $\bU^*$) has a small number of unique rows, then a matrix $\bS$ that samples a small number of rows of $\bA$ has only a small number of possibilities for $\bS\bA$. 

To ensure that $\bA$ has a small number of unique rows for the BMF problem, \cite{KumarPRW19} first find a $2^k$-means clustering solution $\widetilde{\bA}$ for the rows of $\bA$. Instead of solving the problem on $\bA$, they then solve BMF on the matrix $\widetilde{\bA}$, where each row is replaced by the center the point is assigned to, yielding at most $2^k$ unique rows. Finally, they note that $\|\bU^*\bV^* - \bA\|_F^2$ is at least the $2^k$-means cost, as $\bU^*\bV^*$ has at most $2^k$ unique rows. Now that
$\widetilde{\bA}$ has $2^k$ unique rows, they can make all possible guesses for both $\bS\bU^*$ and $\bS\widetilde{\bA}$ in time $2^{\tO{k^2}}$. 
By using an algorithm of \cite{KanungoMNPSW04} that achieves roughly a $9$-approximation to $k$-means clustering, \cite{KumarPRW19} ultimately obtain a $C$-approximation to the BMF problem, for some $C\ge 576$. 

\paragraph{Shortcomings of previous work for $(1+\eps)$-approximation.}
While \cite{KumarPRW19} do not optimize for $C$, their approach fundamentally cannot achieve $(1+\eps)$-approximation for BMF for the following reasons. 
First, they use a $k$-means clustering subroutine~\cite{KanungoMNPSW04}, (achieving roughly a $9$-approximation) which due to hardness-of-approximation results~\cite{Cohen-AddadS19,LeeSW17} can never achieve $(1+\eps)$-approximation, as there cannot exist a $1.07$-approximation algorithm for $k$-means clustering unless P=NP. 
Moreover, even if a $(1+\eps)$-approximate $k$-means clustering could be found, 
there is no guarantee that the cluster centers obtained by this algorithm are binary. 
That is, while $\bU\bV$ has a specific form induced by the requirement that each factor must be binary, a solution to $k$-means clustering offers no such guarantee and may return Steiner points. 
Finally, \cite{KumarPRW19} achieves a matrix $\bS$ that roughly preserves $\bS\bU^*$ and $\bS\bA$.  
By generalizations of the triangle inequality, one can show that $\|\bS\bU^*\bV^*-\bS\bA\|_F^2$ preserves a constant factor approximation to $\|\bU^*\bV^*-\bA\|_F^2$, but not necessarily a $(1+\eps)$-approximation. 

Another related work, \cite{FominGLP020}, reduces instances of BMF to constrained $k$-means clustering instances, where the constraints demand that the selected centers are linear combinations of binary vectors. The core part of their work is to design a sampling-based algorithm for solving binary-constrained clustering instances, and the result on BMF is a corollary. Constrained clustering is a harder problem than BMF with Frobenius loss, so it is unclear how one might improve the doubly exponential running time using this approach.

\paragraph{Our approach: computing a strong coreset.}
We first reduce the number of unique rows in $\bA$ by computing a strong coreset $\widetilde{\bA}$ for $\bA$. 
The strong coreset has the property that for any choices of $\bU\in\{0,1\}^{n\times k}$ and $\bV\in\{0,1\}^{k\times d}$, there exists $\bX\in\{0,1\}^{n\times k}$ such that
\[(1-\eps)\|\bU\bV-\bA\|_F^2\le\|\bX\bV-\widetilde{\bA}\|_F^2\le(1+\eps)\|\bU\bV-\bA\|_F^2.\]
Therefore, we instead first solve the low-rank approximation problem on $\widetilde{\bA}$ first. 
Crucially, we choose $\widetilde{\bA}$ to have $2^{\poly(k/\eps)}$ unique rows so then for a matrix $\bS$ that samples $\poly(k/\eps)$ rows, there are $2^{\poly(k/\eps)}$ possibilities for $\bS\widetilde{\bA}$, so we can make all possible guesses for both $\bS\bU^*$ and $\bS\widetilde{\bA}$. 
Unfortunately, we still have the problem that $\|\bS\bU^*\bV^*-\bS\widetilde{\bA}\|_F^2$ does not even necessarily give a $(1+\eps)$-approximation to $\|\bU^*\bV^*-\widetilde{\bA}\|_F^2$. 

\paragraph{Binary matrix factorization.}
To that end, we show that when $\bS$ is a leverage score sampling matrix, then $\bS$ also satisfies an approximate matrix multiplication property. Therefore $\bS$ can effectively be used for an affine embedding. 
That is, the minimizer to $\|\bS\bU^*\bV^*-\bS\widetilde{\bA}\|_F^2$ produces an $(1+\eps)$-approximation to the cost of the optimal factors $\|\bU^*\bV^*-\widetilde{\bA}\|_F^2$. 
Thus, we can then solve 
\begin{align*}
\bV'&=\argmin_{\bV\in\{0,1\}^{k\times d}}\|\bS\bU^*\bV-\bS\widetilde{\bA}\|_F^2\\
\bU'&=\argmin_{\bU\in\{0,1\}^{n\times k}}\|\bU\bV'-\bA\|_F^2,
\end{align*}
where the latter optimization problem can be solved by iteratively optimizing over each row so that the total computation time is $\O{2^kn}$ rather than $2^{kn}$. 

\paragraph{BMF on binary fields.}
We again form the matrix $\widetilde{\bA}$ by taking a strong coreset of $\bA$, constructed using an algorithm that gives integer weight $w_i$ to each point, and then duplicating the rows to form $\widetilde{\bA}$. 
That is, if the $i$-th row $\bA_i$ of $\bA$ is sampled with weight $w_i$ in the coreset, then $\widetilde{\bA}$ will contain $w_i$ repetitions of the row $\bA_i$.
We want to use the same approach for binary fields to make guesses for $\bS\bU^*$ and $\bS\widetilde{\bA}$. However, it is no longer true that $\bS$ will provide an affine embedding over $\mathbb{F}_2$, in part because the subspace embedding property of $\bS$ computes leverage scores of each row of $\bU^*$ and $\bA$ with respect to general integers. 
Hence, we require a different approach for matrix operations over $\mathbb{F}_2$. 

Instead, we group the rows of $\widetilde{\bA}$ by their number of repetitions, so that group $\bG_j$ consists of the rows of $\widetilde{\bA}$ that are repeated $[(1+\eps)^j,(1+\eps)^{j+1})$ times. 
That is, if $\bA_i$ appears $w_i$ times in $\widetilde{\bA}$, then it appears a single time in group $\bG_j$ for $j=\flr{\log w_i}$. 
We then perform entrywise $L_0$ low-rank approximation over $\mathbb{F}_2$ for each of the groups $\bG_j$, which gives low-rank factors $\bU^{(j)}$ and $\bV^{(j)}$. 
We then compute $\widetilde{\bU^{(j)}}$ by duplicating rows appropriately so that if $\bA_i$ is in $\bG_j$, then we place the row of $\bU^{(j)}$ corresponding to $\bA_i$ into the $i$-th row of $\widetilde{\bU^{(j)}}$, for all $i\in[n]$. 
Otherwise if $\bA_i$ is not in $\bG_j$, then we set $i$-th row of $\widetilde{\bU^{(j)}}$ to be the all zeros row. 
We compute $\bV^{(j)}$ by padding accordingly and then collect 
\[\widetilde{\bU}=\begin{bmatrix}\widetilde{\bU^{(0)}}|\ldots|\widetilde{\bU^{(\ell)}}\end{bmatrix},\qquad\widetilde{\bV}\gets\widetilde{\bV^{(0)}}\circ\ldots\circ\widetilde{\bV^{(i)}},\]
where $\begin{bmatrix}\widetilde{\bU^{(0)}}|\ldots|\widetilde{\bU^{(\ell)}}\end{bmatrix}$ denotes horizontal concatenation of matrices and $\widetilde{\bV^{(0)}}\circ\ldots\circ\widetilde{\bV^{(i)}}$ denotes vertical concatenation (stacking) of matrices,
to achieve bicriteria low-rank approximations $\widetilde{\bU}$ and $\widetilde{\bV}$ to $\widetilde{\bA}$. 
Finally, to achieve bicriteria factors $\bU'$ and $\bV'$ to $\bA$, we ensure that $\bU'$ achieves the same block structure as $\widetilde{\bU}$. 

\paragraph{BMF with $L_p$ loss.}
We would again like to use the same approach as our $(1+\eps)$-approximation algorithm for BMF with Frobenius loss. 
To that end, we observe that a coreset construction for clustering under $L_p$ metrics rather than Euclidean distance is known, which we can use to construct $\widetilde{\bA}$. 
However, the challenge is that no known sampling matrix $\bS$ guarantees an affine embedding. 
One might hope that recent results on active $L_p$ regression~\cite{ChenP19,ParulekarPP21,MuscoMW022,MeyerMMWZ22,MeyerMMWZ23} can provide such a tool. 
Unfortunately, adapting these techniques would still require taking a union bound over a number of columns, which would result in the sampling matrix having too many rows for our desired runtime. 

Instead, we invoke the coreset construction on the rows and the columns so that $\widetilde{\bA}$ has a small number of distinct rows and columns. 
We again partition the rows of $\widetilde{\bA}$ into groups based on their frequency, but now we further partition the groups based on the frequency of the columns. 
Thus, it remains to solve BMF with $L_p$ loss on the partition, each part of which has a small number of rows and columns. 
Since the contribution of each row toward the overall loss is small (because there is a small number of columns), we show that there exists a matrix that samples $\poly(k/\eps)$ rows of each partition that finally achieves the desired affine embedding. 
Therefore, we can solve the problem on each partition, pad the factors accordingly, and build the bicriteria factors as in the binary field case.

\subsection{Motivation and Related Work}
\seclab{sec:related}
Low-rank approximation is one of the fundamental problems of machine learning and data science. 
Therefore, it has received extensive attention, e.g., see the surveys~\cite{KannanV09,Mahoney11,Woodruff14}. 
When the underlying loss function is the Frobenius norm, the low-rank approximation problem can be optimally solved via singular value decomposition (SVD). 
However, when we restrict both the observed input $\bA$ and the factors $\bU,\bV$ to binary matrices, the SVD no longer guarantees optimal factors. 
In fact, many restricted variants of low-rank approximation are NP-hard~\cite{RazenshteynSW16,SongWZ17,KumarPRW19,BanBBKLW19,BanWZ19,FominGLP020,MahankaliW21}.

\paragraph{Motivation and background for BMF.} 
The BMF problem has applications to graph partitioning~\cite{ChandranIK16}, low-density parity-check codes~\cite{RavanbakhshPG16}, and optimizing passive organic LED (OLED) displays~\cite{KumarPRW19}. 
Observe that we can use $\bA$ to encode the incidence matrix of the bipartite graph with $n$ vertices on the left side of the bipartition and $d$ vertices on the right side so that $\bA_{i,j}=1$ if and only if there exists an edge connecting the $i$-th vertex on the left side with the $j$-th vertex on the right side. 
Then $\bU\bV$ can be written as the sum of $k$ rank-$1$ matrices, each encoding a different bipartite clique of the graph, i.e., a subset of vertices on the left and a subset of vertices on the right such that there exists an edge between every vertex on the left and every vertex on the right. 
It then follows that the BMF problem solves the bipartite clique partition problem~\cite{orlin1977contentment,FleischnerMPS09,ChalermsookHHK14,Neumann18}, in which the goal is to find the smallest integer $k$ such that the graph can be represented as a union of $k$ bipartite cliques. 

\cite{KumarPRW19} also present the following motivation for the BMF problem to improve the performance of passive OLED displays, which rapidly and sequentially illuminate rows of lights to render an image in a manner so that the human eye integrates this sequence of lights into a complete image. 
However, \cite{KumarPRW19} observed that passive OLED displays could illuminate many rows simultaneously, provided the image being shown is a rank-$1$ matrix and that the apparent brightness of an image is inversely proportional to the rank of the decomposition. 
Thus \cite{KumarPRW19} notes that BMF can be used to not only find a low-rank decomposition that illuminates pixels in a way that seems brighter to the viewer but also achieves binary restrictions on the decomposition in order to use simple and inexpensive voltage drivers on the rows and columns, rather than a more expensive bank of video-rate digital to analog-to-digital converters.

\paragraph{BMF with Frobenius loss.}
\cite{KumarPRW19} first gave a constant factor approximation algorithm for the BMF problem using runtime $2^{\tO{k^2}}\poly(n,d)$, i.e., singly exponential time.  
\cite{FominGLP020} introduced a $(1+\eps)$-approximation to the BMF problem with rank-$k$ factors, but their algorithm uses doubly exponential time, specifically runtime $2^{\frac{2^{\O{k}}}{\eps^2}\log^2\frac{1}{\eps}}\poly(n,d)$, which was later improved to doubly exponential runtime $2^{\frac{2^{\O{k}}}{\eps^2}\log\frac{1}{\eps}}\poly(n,d)$ by \cite{BanBBKLW19}, who also showed that $2^{k^{\Omega(1)}}$ runtime is necessary even for constant-factor approximation, under the Small Set Expansion Hypothesis and the Exponential Time Hypothesis.
By introducing sparsity constraints on the rows of $\bU$ and $\bV$, \cite{Chen0T022} provide an alternate parametrization of the runtime, though, at the cost of runtime quasipolynomial in $n$ and $d$. 

\paragraph{BMF on binary fields.} 
Binary matrix factorization is particularly suited for datasets involving binary data. Thus, the problem is well-motivated for binary fields when performing dimensionality reduction on high-dimension datasets~\cite{KoyuturkG03}. 
To this end, many heuristics have been developed for this problem~\cite{KoyuturkG03,ShenJY09,fu2010binary,jiang2014clustering}, due to its NP-hardness~\cite{GillisV18,DanHJ0Z18}. 

For the special case of $k=1$, \cite{ShenJY09} first gave a $2$-approximation algorithm that uses polynomial time through a relaxation of integer linear programming. 
Subsequently, \cite{jiang2014clustering} produced a simpler approach, and \cite{BringmannKW17} introduced a sublinear time algorithm.
For general $k$, \cite{KumarPRW19} gave a constant factor approximation algorithm using runtime $2^{\poly(k)}\poly(n,d)$, i.e., singly exponential time, at the expense of a bicriteria solution, i.e., factors with rank $k'=\O{k\log n}$. 
\cite{FominGLP020} introduced a $(1+\eps)$-approximation to the BMF problem with rank-$k$ factors, but their algorithm uses doubly exponential time, specifically runtime $2^{\frac{2^{\O{k}}}{\eps^2}\log^2\frac{1}{\eps}}\poly(n,d)$, which was later improved to doubly exponential runtime $2^{\frac{2^{\O{k}}}{\eps^2}\log\frac{1}{\eps}}\poly(n,d)$ by \cite{BanBBKLW19}, who also showed that doubly exponential runtime is necessary for $(1+\eps)$-approximation without bicriteria relaxation under the Exponential Time Hypothesis.    

\paragraph{BMF with $L_p$ loss.} 
Using more general $L_p$ loss functions can result in drastically different behaviors of the optimal low-rank factors for the BMF problem. 
For example, the low-rank factors for $p>2$ are penalized more when the corresponding entries of $\bU\bV$ are large, and thus may choose to prioritize a larger number of small entries that do not match $\bA$ rather than a single large entry. 
On the other hand, $p=1$ corresponds to robust principal component analysis, which yields factors that are more robust to outliers in the data~\cite{ke2003robust,KeK05,kwak2008principal,ZhengLSYO12,brooks2013pure,MarkopoulosKP14,SongWZ17,park2018three,BanBBKLW19,MahankaliW21}. 
The first approximation algorithm with provable guarantees for $L_1$ low-rank approximation on the reals was given by \cite{SongWZ17}. 
They achieved $\poly(k)\cdot\log d$-approximation in roughly $\O{nd}$ time. 
For constant $k$, \cite{SongWZ17} further achieved constant-factor approximation in polynomial time. 

When we restrict the inputs and factors to be binary, \cite{KumarPRW19} observed that $p=1$ corresponds to minimizing the number of edges in the symmetric difference between an unweighted bipartite graph $G$ and its approximation $H$, which is the multiset union of $k$ bicliques. 
Here we represent the graph $G$ with $n$ and $d$ vertices on the bipartition's left- and right-hand side, respectively, through its edge incidence matrix $\bA$. 
Similarly, we have $\bU_{i,j}=1$ if and only if the $i$-th vertex on the left bipartition is in the $j$-th biclique and $\bV_{i,j}=1$ if and only if the $j$-th vertex on the right bipartition is in the $i$-th biclique. 
Then we have $\|\bU\bV-\bA\|_1=|E(G)\triangle E(H)|$. 
\cite{ChandranIK16} showed how to solve the exact version of the problem, i.e., to recover $\bU,\bV$ under the promise that $\bA=\bU\bV$, using $2^{\O{k^2}}\poly(n,d)$ time. 
\cite{KumarPRW19} recently gave the first constant-factor approximation algorithm for this problem, achieving a $C$-approximation using $2^{\poly(k)}\poly(n,d)$ time, for some constant $C\ge122^{2p-2}+2^{p-1}$. 

\subsection{Preliminaries}
For an integer $n>0$, we use $[n]$ to denote the set $\{1,2,\ldots,n\}$. 
We use $\poly(n)$ to represent a fixed polynomial in $n$ and more generally, $\poly(n_1,\ldots,n_k)$ to represent a fixed multivariate polynomial in $n_1,\ldots,n_k$. 
For a function $f(n_1,\ldots,n_k)$, we use $\tO{f(n_1,\ldots,n_k)}$ to denote $f(n_1,\ldots,n_k)\cdot\poly(\log f(n_1,\ldots,n_k))$. 

We generally use bold-font variables to denote matrices. 
For a matrix $\bA\in\mathbb{R}^{n\times d}$, we use $\bA_i$ to denote the $i$-th row of $\bA$ and $\bA^{(j)}$ to denote the $j$-th column of $\bA$. 
We use $A_{i,j}$ to denote the entry in the $i$-th row and $j$-th column of $\bA$. 
For $p\ge 1$, we write the entrywise $L_p$ norm of $\bA$ as
\[\|\bA\|_p=\left(\sum_{i\in[n]}\sum_{j\in[d]}A_{i,j}^p\right)^{1/p}.\]
The Frobenius norm of $\bA$, denoted $\|\bA\|_F$ is simply the entrywise $L_2$ norm of $\bA$:
\[\|\bA\|_F=\left(\sum_{i\in[n]}\sum_{j\in[d]}A_{i,j}^2\right)^{1/2}.\]
The entrywise $L_0$ norm of $\bA$ is 
\[\|\bA\|_0=\left\lvert\{(i,j)\,\mid\,i\in[n], j\in[d]: A_{i,j}\neq 0\}\right\rvert.\]
We use $\circ$ to denote vertical stacking of matrices, so that \[\bA^{(1)}\circ\ldots\circ\bA^{(m)}=\begin{bmatrix}\bA^{(1)}\\\vdots\\\bA^{(m)}\end{bmatrix}.\]

For a set $X$ of $n$ points in $\mathbb{R}^d$ weighted by a function $w$, the $k$-means clustering cost of $X$ with respect to a set $S$ of $k$ centers is defined as
\[\Cost(X,S,w):=\sum_{x\in X}w(x)\cdot\min_{s\in S}\|x-s\|_2^2.\]
When the weights $w$ are uniformly unit across all points in $X$, we simply write $\Cost(X,S)=\Cost(X,S,w)$. 

One of the core ingredients for avoiding the triangle inequality and achieving $(1+\eps)$-approximation is our use of coresets for $k$-means clustering:

\begin{definition}[Strong coreset]
Given an accuracy parameter $\eps>0$ and a set $X$ of $n$ points in $\mathbb{R}^d$, we say that a subset $C$ of $X$ with weights $w$ is a strong $\eps$-coreset of $X$ for the $k$-means clustering problem if for any set $S$ of $k$ points in $\mathbb{R}^d$, we have 
\[(1-\eps)\Cost(X,S)\le\Cost(C,S,w)\le(1+\eps)\Cost(X,S).\]
\end{definition}

Many coreset construction exist in the literature, and the goal is to minimize $|C|$, the size of the coreset, while preserving $(1\pm\eps)$-approximate cost for all sets of $k$ centers. If the points lie in $\mathbb{R}^d$, we can find coresets of size $\tO{\poly(k,d,\epsilon^{-1})}$, i.e., the size is independent of $n$.

\paragraph{Leverage scores.} 
Finally, we recall the notion of a leverage score sampling matrix. 
For a matrix $\bA\in\mathbb{R}^{n\times d}$, the leverage score of row $\ba_i$ with $i\in[n]$ is defined as $\ba_i(\bA^\top\bA)^{-1}\ba_i^\top$. 
We can use the leverage scores to generate a random leverage score sampling matrix as follows:
\begin{theorem}[Leverage score sampling matrix]
\thmlab{thm:lev:score:sample}
\cite{DrineasMM06a,DrineasMM06b,MagdonIsmail10,Woodruff14}
Let $C>1$ be a universal constant and $\alpha>1$ be a parameter. 
Given a matrix $\bA\in\mathbb{R}^{n\times d}$, let $\ell_i$ be the leverage score of the $i$-th row of $\bA$. 
Suppose $p_i\in\left[\min\left(1,\frac{C\ell_i\log k}{\eps^2}\right),\min\left(1,\frac{C\alpha\ell_i\log k}{\eps^2}\right)\right]$ for all $i\in[n]$. 

For $m:=\O{\frac{\alpha}{\eps^2}\,d\log d}$, let $\bS\in\mathbb{R}^{m\times n}$ be generated so that each row of $\bS$ randomly selects row $j\in[n]$ with probability proportional to $p_j$ and rescales the row by $\frac{1}{\sqrt{mp_j}}$.  
Then with probability at least $0.99$, we have that simultaneously for all vectors $\bx\in\mathbb{R}^d$,
\[(1-\eps)\|\bA\bx\|_2\le\|\bS\bA\bx\|_2\le(1+\eps)\|\bA\bx\|_2.\]
\end{theorem}
The main point of \thmref{thm:lev:score:sample} is that given constant-factor approximations $p_i$ to the leverage scores $\ell_i$, it suffices to sample $\O{d\log d}$ rows of $\bA$ to achieve a constant-factor subspace embedding of $\bA$, and similar bounds can be achieved for $(1+\eps)$-approximate subspace embeddings. 
Finally, we remark that $\bS$ can be decomposed as the product of matrices $\bD\bT$, where $\bT\in\mathbb{R}^{m\times n}$ is a sparse matrix with a single one per row, denoting the selection of a row for the purposes of leverage score sampling and $\bD$ is the diagonal matrix with the corresponding scaling factor, i.e., the $i$-th diagonal entry of $\bD$ is set to $\frac{1}{\sqrt{mp_j}}$ if the $j$-th row of $\bA$ is selected for the $i$-th sample. 

\section{Binary Low-Rank Approximation}
\seclab{sec:binary:lra}
In this section, we present a $(1+\eps)$-approximation algorithm for binary low-rank approximation with Frobenius norm loss, where to goal is to find matrices $\bU\in\{0,1\}^{n\times k}$ and $\bV\in\{0,1\}^{k\times d}$ to minimize $\|\bU\bV-\bA\|_F^2$. 
Suppose optimal low-rank factors are $\bU^*\in\{0,1\}^{n\times k}$ and $\bV^*\in\{0,1\}^{k\times d}$, so that
\[\|\bU^*\bV^*-\bA\|_F^2=\min_{\bU\in\{0,1\}^{n\times k},\bV\in\{0,1\}^{k\times d}}\|\bU\bV-\bA\|_F^2.\]
Observe that if we knew matrices $\bS\bU^*$ and $\bS\bA$ so that for all $\bV\in\{0,1\}^{k\times d}$,
\[(1-\eps)\|\bU^*\bV-\bA\|_F^2\le\|\bS\bU^*\bV-\bS\bA\|_F^2\le(1+\eps)\|\bU^*\bV-\bA\|_F^2,\]
then we could find a $(1+\eps)$-approximate solution for $\bV^*$ by solving the problem 
\[\argmin_{\bV\in\{0,1\}^{k\times d}}\|\bS\bU^*\bV-\bS\bA\|_F^2\]
instead. 

We would like to make guesses for the matrices $\bS\bU^*$ and $\bS\bA$, but first we must ensure there are not too many possibilities for these matrices. 
For example, if we chose $\bS$ to be a dense matrix with random gaussian entries, then $\bS\bU^*$ could have too many possibilities because without additional information, there are $2^{nk}$ possibilities for the matrix $\bU^*\in\{0,1\}^{n\times k}$. 
We can instead choose $\bS$ to be a leverage score sampling matrix, which samples rows from $\bU^*$ and $\bA$. 
Since each row of $\bU^*$ has dimension $k$, then there are at most $2^k$ distinct possibilities for each of the rows of $\bU^*$. 
On the other hand, $\bA\in\{0,1\}^{n\times d}$, so there may be $2^d$ distinct possibilities for the rows of $\bA$, which is too many to guess. 

Thus we first reduce the number of unique rows in $\bA$ by computing a strong coreset $\widetilde{\bA}$ for $\bA$. 
The strong coreset has the property that for any choices of $\bU\in\{0,1\}^{n\times k}$ and $\bV\in\{0,1\}^{k\times d}$, there exists $\bX\in\{0,1\}^{n\times k}$ such that
\[(1-\eps)\|\bU\bV-\bA\|_F^2\le\|\bX\bV-\widetilde{\bA}\|_F^2\le(1+\eps)\|\bU\bV-\bA\|_F^2.\]
Therefore, we instead first solve the low-rank approximation problem on $\widetilde{\bA}$ first. 
Crucially, $\widetilde{\bA}$ has $2^{\poly(k/\eps)}$ unique rows so then for a matrix $\bS$ that samples $\poly(k/\eps)$ rows, there are $\binom{2^{\poly(k/\eps)}}{\poly(k /\eps)} = 2^{\poly(k/\eps)}$ possible choices of $\bS\widetilde{\bA}$, so we can enumerate all of them for both $\bS\bU^*$ and $\bS\widetilde{\bA}$. 
We can then solve 
\[\bV'=\argmin_{\bV\in\{0,1\}^{k\times d}}\|\bS\bU^*\bV-\bS\widetilde{\bA}\|_F^2\]
and
\[\bU'=\argmin_{\bU\in\{0,1\}^{n\times k}}\|\bU\bV'-\bA\|_F^2,\]
where the latter optimization problem can be solved by iteratively optimizing over each row, so that the total computation time is $\O{2^kn}$ rather than $2^{kn}$. 
We give the full algorithm in \algref{alg:lra} and the subroutine for optimizing with respect to $\widetilde{\bA}$ in \algref{alg:distinct:lra}. 
We give the subroutines for solving for $\bV'$ and $\bU'$ in \algref{alg:compute:v} and \algref{alg:compute:u}, respectively. 

\begin{algorithm}[!htb]
\caption{Algorithm for computing optimal $\bV$ given $\bU$}
\alglab{alg:compute:v}
\begin{algorithmic}[1]
\Require{$\widetilde{\bA}\in\{0,1\}^{N\times d}$, $\bU\in\{0,1\}^{N\times k}$}
\Ensure{$\bV'=\argmin_{\bV\in\{0,1\}^{k\times d}}\|\bU\bV-\widetilde{\bA}\|_F$}
\For{$i=1$ to $i=d$}
\Comment{Optimize for each column individually}
\State{Set $\bV'^{(i)}=\argmin_{\bV^{(i)}\in\{0,1\}^{k\times 1}}\|\bU\bV^{(i)}-\widetilde{\bA}^{(i)}\|_2$}
\Comment{Enumerate over all $2^k$ possible binary vectors}
\EndFor
\State{\Return $\bV'=\begin{bmatrix}\bV'^{(1)}|\ldots|\bV'^{(d)}\end{bmatrix}$}
\end{algorithmic}
\end{algorithm}

\begin{algorithm}[!htb]
\caption{Algorithm for computing optimal $\bU$ given $\bV$}
\alglab{alg:compute:u}
\begin{algorithmic}[1]
\Require{$\widetilde{\bA}\in\{0,1\}^{N\times d}$, $\bV\in\{0,1\}^{k\times d}$}
\Ensure{$\bU'=\argmin_{\bU\in\{0,1\}^{N\times k}}\|\bU\bV-\widetilde{\bA}\|_F$}
\For{$i=1$ to $i=N$}
\Comment{Optimize for each row individually}
\State{Set $\bU'_i=\argmin_{\bU_i\in\{0,1\}^{1\times k}}\|\bU_i\bV-\widetilde{\bA}_i\|_2$}
\Comment{Enumerate over all $2^k$ possible binary vectors}
\EndFor
\State{\Return $\bU'=\bU'_1\circ\ldots\circ\bU'_N$}
\end{algorithmic}
\end{algorithm}

\begin{algorithm}[!htb]
\caption{Low-rank approximation for matrix $\widetilde{\bA}$ with $t$ distinct rows}
\alglab{alg:distinct:lra}
\begin{algorithmic}[1]
\Require{$\widetilde{\bA}\in\{0,1\}^{N\times d}$ with at most $t$ distinct rows, rank parameter $k$, accuracy parameter $\eps>0$}
\Ensure{$\bU'\in\{0,1\}^{n\times k}, \bV'\in\{0,1\}^{k\times d}$ satisfying the property that $\|\bU'\bV'-\bA\|_F^2\le(1+\eps)\min_{\bU\in\{0,1\}^{n\times k}, \bV\in\{0,1\}^{k\times d}}\|\bU\bV-\widetilde{\bA}\|_F^2$}
\State{$V\gets\emptyset$}
\For{each guess of $\bS\bU^*$ and $\bS\widetilde{\bA}$, where $\bS$ is a leverage score sampling matrix with $m=\O{\frac{k\log k}{\eps^2}}$ rows with weights that are powers of two up to $\poly(N)$}
\State{$V\gets V\cup\argmin_{\bV\in\{0,1\}^{k\times d}}\|\bS\bU^*\bV-\bS\widetilde{\bA}\|_F^2$}
\Comment{\algref{alg:compute:v}}
\EndFor
\For{each $\bV\in V$}
\State{Let $\bU_\bV=\argmin_{\bU\in\{0,1\}^{N\times k}}\|\bU\bV-\widetilde{\bA}\|_F^2$}
\Comment{\algref{alg:compute:u}}
\EndFor
\State{$\bV'\gets\argmin_{\bV\in\{0,1\}^{k\times d}}\|\bS\bU_\bV\bV-\bS\widetilde{\bA}\|_F^2$}
\State{$\bU'\gets\bU_{\bV'}$}
\State{\Return $(\bU',\bV')$}
\end{algorithmic}
\end{algorithm}

First, we recall that leverage score sampling matrices preserve approximate matrix multiplication. 
\begin{lemma}[Lemma 32 in \cite{ClarksonW13}]
\lemlab{lem:amm}
Let $\bU\in\mathbb{R}^{N\times k}$ have orthonormal columns, $\widetilde{\bA}\in\{0,1\}^{N\times d}$, and $\bS\in\mathbb{R}^{m\times N}$ be a leverage score sampling matrix for $\bU$ with $m=\O{\frac{1}{\eps^2}}$ rows. 
Then,
\[\PPr{\|\bU^\top\bS^\top\bS\widetilde{\bA}-\bU^\top\widetilde{\bA}\|_F^2<\eps^2\|\bU\|_F^2\|\widetilde{\bA}\|_F^2}\ge 0.99.\]
\end{lemma}

Next, we recall that leverage score sampling matrices give subspace embeddings. 
\begin{theorem}[Theorem 42 in \cite{ClarksonW13}]
\thmlab{thm:se}
For $\bU\in\mathbb{R}^{N\times k}$, let $\bS\in\mathbb{R}^{m\times N}$ be a leverage score sampling matrix for $\bU\in\{0,1\}^{N\times k}$ with $m=\O{\frac{k\log k}{\eps^2}}$ rows. 
Then with probability at least $0.99$, we have for all $\bV\in\mathbb{R}^{k\times d}$,
\[(1-\eps)\|\bU\bV\|_F^2\le\|\bS\bU\bV\|_F^2\le(1+\eps)\|\bU\bV\|_F^2.\]
\end{theorem}

Finally, we recall that approximate matrix multiplication and leverage score sampling suffices to achieve an affine embedding.
\begin{theorem}[Theorem 39 in \cite{ClarksonW13}]
\thmlab{thm:affine}
Let $\bU\in\mathbb{R}^{N\times k}$ have orthonormal columns. 
Let $\bS$ be a sampling matrix that satisfies \lemref{lem:amm} with error parameter $\frac{\eps}{\sqrt{k}}$ and also let $\bS$ be a subspace embedding for $\bU$ with error parameter $\eps$. 
Let $\bV^*=\argmin_{\bV}\|\bU\bV-\widetilde{\bA}\|_F$ and $\bX=\bU\bV^*-\widetilde{\bA}$. 
Then for all $\bV\in\mathbb{R}^{k\times d}$,
\[(1-2\eps)\|\bU\bV-\widetilde{\bA}\|_F^2-\|\bX\|_F^2\le\|\bS\bU\bV-\bS\widetilde{\bA}\|_F^2-\|\bS\bX\|_F^2\le(1+2\eps)\|\bU\bV-\widetilde{\bA}\|_F^2-\|\bX\|_F^2.\]
\end{theorem}

We first show that \algref{alg:distinct:lra} achieves a good approximation to the optimal low-rank factors for the coreset $\widetilde{\bA}$. 
\begin{lemma}
\lemlab{lem:correctness:subroutine}
Suppose $\eps<\frac{1}{10}$. 
Then with probability at least $0.97$, the output of \algref{alg:distinct:lra} satisfies 
\[\|\bU'\bV'-\widetilde{\bA}\|_F^2\le(1+6\eps)\|\bU^*\bV^*-\widetilde{\bA}\|_F^2.\]
\end{lemma}
\begin{proof}
Let $\bV''=\argmin_{\bV\in\{0,1\}^{k\times d}}\|\bS\bU^*\bV-\widetilde{\bA}\|_F^2$ and let $\bU''=\argmin_{\bU\in\{0,1\}^{N\times k}}\|\bS\bU\bV''-\widetilde{\bA}\|_F^2$
Since the algorithm chooses $\bU'$ and $\bV'$ over $\bU''$ and $\bV''$, then
\[\|\bU'\bV'-\widetilde{\bA}\|_F^2\le\|\bU''\bV''-\widetilde{\bA}\|_F^2.\]
Due to the optimality of $\bU''$,
\[\|\bU''\bV''-\widetilde{\bA}\|_F^2\le\|\bU^*\bV''-\widetilde{\bA}\|_F^2.\]
Let $\bX=\bU^*\bV^*-\widetilde{\bA}$. 
Note that since $\bU^*$ has orthonormal columns, then by \lemref{lem:amm}, the leverage score sampling matrix $\bS$ achieves approximate matrix multiplication with probability at least $0.99$. 
By \thmref{thm:se}, the matrix $\bS$ also is a subspace embedding for $\bU$. 
Thus, $\bS$ meets the criteria for applying \thmref{thm:affine}. 
Then for the correct guess $\bD\bT$ of matrix $\bS$ corresponding to $\bU^*$ and conditioning on the correctness of $\bS$ in \thmref{thm:affine}, 
\[\|\bU^*\bV''-\widetilde{\bA}\|_F^2\le\frac{1}{1-2\eps}[\|\bS\bU^*\bV''-\bS\widetilde{\bA}\|_F^2-\|\bS\bX\|_F^2+\|\bX\|_F^2.]\]
Due to the optimality of $\bV''$,
\[\frac{1}{1-2\eps}[\|\bS\bU^*\bV''-\bS\widetilde{\bA}\|_F^2-\|\bS\bX\|_F^2+\|\bX\|_F^2]\le\frac{1}{1-2\eps}[\|\bS\bU^*\bV^*-\bS\widetilde{\bA}\|_F^2-\|\bS\bX\|_F^2+\|\bX\|_F^2].\]
Then again conditioning on the correctness of $\bS$,
\begin{align*}
\frac{1}{1-2\eps}[\|\bS\bU^*\bV^*&-\bS\widetilde{\bA}\|_F^2-\|\bS\bX\|_F^2+\|\bX\|_F^2]\\
&\le\frac{1}{1-2\eps}[(1+2\eps)\|\bU^*\bV^*-\widetilde{\bA}\|_F^2+\|\bS\bX\|_F^2-\|\bX\|_F^2-\|\bS\bX\|_F^2+\|\bX\|_F^2]\\
&\le(1+6\eps)\|\bU^*\bV^*-\widetilde{\bA}\|_F^2,
\end{align*}
for sufficiently small $\eps$, e.g., $\eps<\frac{1}{10}$. 
Thus, putting things together, we have that conditioned on the correctness of $\bS$ in \thmref{thm:affine},
\[\|\bU'\bV'-\widetilde{\bA}\|_F^2\le(1+6\eps)\|\bU^*\bV^*-\widetilde{\bA}\|_F^2.\]
Since the approximate matrix multiplication property of \lemref{lem:amm}, the subspace embedding property of \thmref{thm:se}, and the affine embedding property of \thmref{thm:affine} all fail with probability at most $0.01$, then by a union bound, $\bS$ succeeds with probability at least $0.97$. 
\end{proof}

We now analyze the runtime of the subroutine \algref{alg:distinct:lra}. 
\begin{lemma}
\lemlab{lem:runtime:subroutine}
\algref{alg:distinct:lra} uses $2^{\O{m^2+m\log t}}\poly(N,d)$ runtime for $m=\O{\frac{k\log k}{\eps^2}}$. 
\end{lemma}
\begin{proof}
We analyze the number of possible guesses $\bD$ and $\bT$ corresponding to guesses of $\bS\widetilde{\bA}$ (see the remark after \thmref{thm:lev:score:sample}). 
There are at most $\binom{t}{m}=2^{\O{m\log t}}$ distinct subsets of $m=\O{\frac{k\log k}{\eps^2}}$ rows of $\widetilde{\bA}$. 
Thus there are $2^{\O{m\log t}}$ possible matrices $\bT$ that selects $m$ rows of $\widetilde{\bA}$, for the purposes of leverage score sampling.  
Assuming the leverage score sampling matrix does not sample any rows with leverage score less than $\frac{1}{\poly(N)}$, then there are $\O{\log N}^m=2^{\O{m\log\log N}}$ total guesses for the matrix $\bD$. 
Note that $\log n\le 2^m$ implies that $2^{\O{m\log\log N}}\le 2^{\O{m^2}}$ while $\log N>2^m$ implies that $2^{\O{m\log\log N}}\le2^{\O{\log^2\log N}}\le N$. 
Therefore, there are at most $2^{\O{m^2+m\log t}} N$ total guesses for all combinations of $\bT$ and $\bD$, corresponding to all guesses of $\bS\widetilde{\bA}$. 

For each guess of $\bS$ and $\bS\widetilde{\bA}$, we also need to guess $\bS\bU^*$. 
Since $\bU^*\in\{0,1\}^{N\times k}$ is binary and $\bT$ samples $m$ rows before weighting each row with one of $\O{\log N}$ possible weights, the number of total guesses for $\bS\bU^*$ is $(2\cdot\O{\log N})^{mk}$. 

Given guesses for $\bS\bA$ and $\bS\bU^*$, we can then compute $\argmin_{\bV\in\{0,1\}^{k\times d}}\|\bS\bU^*\bV-\bS\bA\|_F^2$ using $\O{2^k d}$ time through the subroutine \algref{alg:compute:v}, which enumerates through all possible $2^k$ binary vectors for each column. 
For a fixed $\bV$, we can then compute $\bU_\bV=\argmin_{\bU\in\{0,1\}^{N\times k}}\|\bU\bV-\bA\|_F^2$ using $\O{2^k N}$ time through the subroutine \algref{alg:compute:u}, which enumerates through all possible $2^k$ binary vectors for each row of $\bU_\bV$. 
Therefore, the total runtime of \algref{alg:distinct:lra} is $2^{\O{m^2+m\log t}}\poly(N,d)$. 
\end{proof}



We recall the following construction for a strong $\eps$-coreset for $k$-means clustering. 
\begin{theorem}[Theorem 36 in \cite{FeldmanSS20}]
\thmlab{thm:strong:coreset}
Let $X\subset\mathbb{R}^d$ be a subset of $n$ points, $\eps\in(0,1)$ be an accuracy parameter, and let $t=\O{\frac{k^3\log^2 k}{\eps^4}}$. 
There exists an algorithm that uses $\O{nd^2+n^2d+\frac{nkd}{\eps^2}+\frac{nk^2}{\eps^2}}$ time and outputs a set of $t$ weighted points that is a strong $\eps$-coreset for $k$-means clustering with probability at least $0.99$. 
Moreover, each point has an integer weight that is at most $\poly(n)$. 
\end{theorem}

\begin{algorithm}[!htb]
\caption{Low-rank approximation for matrix $\bA$}
\alglab{alg:lra}
\begin{algorithmic}[1]
\Require{$\bA\in\{0,1\}^{n\times d}$, rank parameter $k$, accuracy parameter $\eps>0$}
\Ensure{$\bU'\in\{0,1\}^{n\times k}, \bV'\in\{0,1\}^{k\times d}$ satisfying the property that $\|\bU'\bV'-\bA\|_F^2\le(1+\eps)\min_{\bU\in\{0,1\}^{n\times k}, \bV\in\{0,1\}^{k\times d}}\|\bU\bV-\bA\|_F^2$}
\State{$t\gets\O{\frac{2^{3k}k^2}{\eps^4}}$}
\Comment{\thmref{thm:strong:coreset} for $2^k$-means clustering}
\State{Compute a strong coreset $C$ for $2^k$-means clustering of $\bA$, with size $t$ and total weight $N=\poly(n)$}
\State{Let $\widetilde{\bA}\in\{0,1\}^{N\times d}$ be the matrix representation of $C$, where weighted points are duplicated appropriately}
\State{Let $(\widetilde{\bU},\widetilde{\bV})$ be the output of \algref{alg:distinct:lra} on input $\widetilde{\bA}$}
\State{$\bU'\gets\argmin_{\bU\in\{0,1\}^{n\times k}}\|\bU\widetilde{\bV}-\bA\|_F^2$, $\bV'\gets\widetilde{\bV}$}
\Comment{\algref{alg:compute:u}}
\State{\Return $(\bU',\bV')$}
\end{algorithmic}
\end{algorithm}

We now justify the correctness of \algref{alg:lra}. 
\begin{lemma}
\lemlab{lem:correctness:lra}
With probability at least $0.95$, \algref{alg:lra} returns $\bU',\bV'$ such that
\[\|\bU'\bV'-\bA\|_F^2\le(1+\eps)\min_{\bU\in\{0,1\}^{n\times k},\bV\in\{0,1\}^{k\times d}}\|\bU\bV-\bA\|_F^2.\]
\end{lemma}
\begin{proof}
Let $\widetilde{\bM}$ be the indicator matrix that selects a row of $\widetilde{\bU}\widetilde{\bV}=\widetilde{\bU}\bV'$ to match to each row of $\bA$, so that by the optimality of $\bU'$, 
\[\|\bU'\bV'-\bA\|_F^2\le\|\widetilde{\bM}\widetilde{\bU}\widetilde{\bV}-\bA\|_F^2.\]

Note that any $\bV$ is a set of $k$ points in $\{0,1\}^d$ and so each row $\bU_i$ of $\bU$ induces one of at most $2^k$ possible points $\bU_i\bV\in\{0,1\}^d$. 
Hence $\|\bU\bV-\bA\|_F^2$ is the objective value of a constrained $2^k$-means clustering problem. 
Thus by the choice of $t$ in \thmref{thm:strong:coreset}, we have that $\widetilde{\bA}$ is a strong coreset, so that
\[\|\widetilde{\bM}\widetilde{\bU}\widetilde{\bV}-\bA\|_F^2\le(1+\eps)\|\widetilde{\bU}\widetilde{\bV}-\widetilde{\bA}\|_F^2.\] 
Let $\bU^*\in\{0,1\}^{n\times k}$ and $\bV^*\in\{0,1\}^{k\times d}$ such that 
\[\|\bU^*\bV^*-\bA\|_F^2=\min_{\bU\in\{0,1\}^{n\times k},\bV\in\{0,1\}^{k\times d}}\|\bU\bV-\bA\|_F^2.\]
Let $\bM^*$ be the indicator matrix that selects a row of $\bU^*\bV^*$ to match to each row of $\widetilde{\bA}$, so that by \lemref{lem:correctness:subroutine}, 
\[(1+\eps)\|\widetilde{\bU}\widetilde{\bV}-\widetilde{\bA}\|_F^2\le(1+\eps)^2\|\bM^*\bU^*\bV^*-\widetilde{\bA}\|_F^2.\]
Then by the choice of $t$ in \thmref{thm:strong:coreset}, we have that
\[(1+\eps)^2\|\bM^*\bU^*\bV^*-\widetilde{\bA}\|_F^2\le(1+\eps)^3\|\bU^*\bV^*-\bA\|_F^2.\]
The desired claim then follows from rescaling $\eps$. 
\end{proof}
We now analyze the runtime of \algref{alg:lra}. 
\begin{lemma}
\lemlab{lem:runtime:lra}
\algref{alg:lra} uses $2^{\tO{k^2/\eps^4}}\poly(n,d)$ runtime. 
\end{lemma}
\begin{proof}
By \thmref{thm:strong:coreset}, it follows that \algref{alg:lra} uses $\O{nd^2+n^2d+\frac{nkd}{\eps^2}+\frac{nk^2}{\eps^2}}$ time to compute $\widetilde{\bA}\in\{0,1\}^{N\times d}$ with $N=\poly(n)$. 
By \lemref{lem:runtime:subroutine}, it follows that \algref{alg:distinct:lra} on input $\widetilde{\bA}$ thus uses runtime $2^{\O{m^2+m\log t}}\poly(N,d)$ for $m=\O{\frac{k\log k}{\eps^2}}$ and $t=\O{\frac{2^{3k}k^2}{\eps^4}}$. 
Finally, computing $\bU'$ via \algref{alg:compute:u} takes $\O{2^k n}$ time after enumerating through all possible $2^k$ binary vectors for each row of $\bU'$. 
Therefore, the total runtime of \algref{alg:lra} is $2^{\tO{\frac{k^2\log^2 k}{\eps^4}}}\poly(n,d)=2^{\tO{k^2/\eps^4}}\poly(n,d)$.
\end{proof}

Combining \lemref{lem:correctness:lra} and \lemref{lem:runtime:lra}, we have:
\begin{theorem}
There exists an algorithm that uses $2^{\tO{k^2/\eps^4}}\poly(n,d)$ runtime and with probability at least $\frac{2}{3}$, outputs $\bU'\in\{0,1\}^{n\times k}$ and $\bV'\in\{0,1\}^{k\times d}$ such that
\[\|\bU'\bV'-\bA\|_F^2\le(1+\eps)\min_{\bU\in\{0,1\}^{n\times k},\bV\in\{0,1\}^{k\times d}}\|\bU\bV-\bA\|_F^2.\]
\end{theorem}

\section{\texorpdfstring{$\mathbb{F}_2$}{GF2} Low-Rank Approximation}
\seclab{sec:lra:frob:ftwo}
In this section, we present a $(1+\eps)$-approximation algorithm for binary low-rank approximation on $\mathbb{F}_2$, where to goal is to find matrices $\bU\in\{0,1\}^{n\times k}$ and $\bV\in\{0,1\}^{k\times d}$ to minimize the Frobenius norm loss $\|\bU\bV-\bA\|_F^2$, but now all operations are performed in $\mathbb{F}_2$. 
We would like to use the same approach as in \secref{sec:binary:lra}, i.e., to make guesses for the matrices $\bS\bU^*$ and $\bS\bA$ while ensuring there are not too many possibilities for these matrices. 
To do so for matrix operations over general integers, we chose $\bS$ to be a leverage score sampling matrix that samples rows from $\bU^*$ and $\bA$. 
We then used the approximate matrix multiplication property in \lemref{lem:amm} and the subspace embedding property in \thmref{thm:se} to show that $\bS$ provides an affine embedding in \thmref{thm:affine} over general integers. 
However, it no longer necessarily seems true that $\bS$ will provide an affine embedding over $\mathbb{F}_2$, in part because the subspace embedding property of $\bS$ computes leverage scores of each row of $\bU^*$ and $\bA$ with respect to general integers. 
Thus we require an alternate approach for matrix operations over $\mathbb{F}_2$. 

Instead, we form the matrix $\widetilde{\bA}$ by taking a strong coreset of $\bA$ and then duplicating the rows according to their weight $w_i$ to form $\widetilde{\bA}$. 
That is, if the $i$-th row $\bA_i$ of $\bA$ is sampled with weight $w_i$ in the coreset, then $\widetilde{\bA}$ will contain $w_i$ repetitions of the row $\bA_i$, where we note that $w_i$ is an integer. 
We then group the rows of $\widetilde{\bA}$ by their repetitions, so that group $\bG_j$ consists of the rows of $\widetilde{\bA}$ that are repeated $[(1+\eps)^j,(1+\eps)^{j+1})$ times. 
Thus if $\bA_i$ appears $w_i$ times in $\widetilde{\bA}$, then it appears a single time in group $\bG_j$ for $j=\flr{\log w_i}$. 

We perform entrywise $L_0$ low-rank approximation over $\mathbb{F}_2$ for each of the groups $\bG_j$, which gives low-rank factors $\bU^{(j)}$ and $\bV^{(j)}$. 
We then compute $\widetilde{\bU^{(j)}}\in\mathbb{R}^{n\times d}$ from $\bU^{(j)}$ by following procedure. 
If $\bA_i$ is in $\bG_j$, then we place the row of $\bU^{(j)}$ corresponding to $\bA_i$ into the $i$-th row of $\widetilde{\bU^{(j)}}$, for all $i\in[n]$. 
Note that the row of $\bU^{(j)}$ corresponding to $\bA_i$ may not be the $i$-th row of $\bU^{(j)}$, e.g., since $\bA_i$ will appear only once in $\bG_j$ even though it appears $w_i \in [(1+\eps)^j,(1+\eps)^{j+1})$ times in $\bA$. 
Otherwise if $\bA_i$ is not in $\bG_j$, then we set $i$-th row of $\widetilde{\bU^{(j)}}$ to be the all zeros row. 
We then achieve $\bV^{(j)}$ by padding accordingly. 
Finally, we collect \[\widetilde{\bU}=\begin{bmatrix}\widetilde{\bU^{(0)}}|\ldots|\widetilde{\bU^{(\ell)}}\end{bmatrix},\qquad\widetilde{\bV}\gets\widetilde{\bV^{(0)}}\circ\ldots\circ\widetilde{\bV^{(i)}}\]
to achieve bicriteria low-rank approximations $\widetilde{\bU}$ and $\widetilde{\bV}$ to $\widetilde{\bA}$. 
Finally, to achieve bicriteria low-rank approximations $\bU'$ and $\bV'$ to $\bA$, we require that $\bU'$ achieves the same block structure as $\widetilde{\bU}$. 
We describe this subroutine in \algref{alg:compute:u:gf2} and we give the full low-rank approximation bicriteria algorithm in \algref{alg:lra:gf2}. 

We first recall the following subroutine to achieve entrywise $L_0$ low-rank approximation over $\mathbb{F}_2$. 
Note that for matrix operations over $\mathbb{F}_2$, we have that the entrywise $L_0$ norm is the same as the entrywise $L_p$ norm for all $p$. 
\begin{lemma}[Theorem 3 in \cite{BanBBKLW19}]
\lemlab{lem:l0:sample:matrix}
For $\eps\in(0,1)$, there exists a $(1+\eps$)-approximation algorithm to entrywise $L_0$ rank-$k$ approximation over $\mathbb{F}_2$ running in $d\cdot n^{\poly(k/\eps)}$ time. 
\end{lemma}

\begin{algorithm}[!htb]
\caption{Algorithm for computing optimal $\bU$ given $\bV^{(1)},\ldots,\bV^{(\ell)}$}
\alglab{alg:compute:u:gf2}
\begin{algorithmic}[1]
\Require{$\widetilde{\bA}\in\{0,1\}^{N\times d}$, $\bV^{(1)},\ldots,\bV^{(\ell)}\in\{0,1\}^{k\times d}$}
\Ensure{$\bU'=\argmin_{\bU\in\{0,1\}^{N\times \ell k}}\|\bU\bV-\widetilde{\bA}\|_F$, where $\bU$ is restricted to one nonzero block of $k$ coordinates}
\For{$i=1$ to $i=N$}
\State{Set $(\bU'_i,j')=\argmin_{\bU_i\in\{0,1\}^{1\times k},j\in[\ell]}\|\bU_i\bV^{(j)}-\widetilde{\bA}_i\|_2$}
\Comment{Enumerate over all $2^k$ possible binary vectors, all $\ell$ indices}
\State{Pad $\bU'_i$ with length $\ell k$, as the $j'$-th block of $k$ coordinates}
\EndFor
\State{\Return $\bU'=\bU'_1\circ\ldots\circ\bU'_N$}
\end{algorithmic}
\end{algorithm}

\begin{algorithm}[!htb]
\caption{Bicriteria low-rank approximation on $\mathbb{F}_2$ for matrix $\bA$}
\alglab{alg:lra:gf2}
\begin{algorithmic}[1]
\Require{$\bA\in\{0,1\}^{n\times d}$, rank parameter $k$, accuracy parameter $\eps>0$}
\Ensure{$\bU'\in\{0,1\}^{n\times k}, \bV'\in\{0,1\}^{k\times d}$ satisfying the property that $\|\bU'\bV'-\bA\|_F^2\le(1+\eps)\min_{\bU\in\{0,1\}^{n\times k}, \bV\in\{0,1\}^{k\times d}}\|\bU\bV-\bA\|_F^2$, where all matrix operations are performed in $\mathbb{F}_2$}
\State{$\ell\gets\O{\frac{\log n}{\eps}}$, $t\gets\O{\frac{(2^k\ell)^3k^2}{\eps^4}}$, $k'\gets\ell k$}
\Comment{\thmref{thm:strong:coreset} for $2^k$-means clustering}
\State{Compute a strong coreset $C$ for $2^k$-means clustering of $\bA$, with size $t$ and total weight $N=\poly(n)$}
\State{Let $\widetilde{\bA}\in\{0,1\}^{N\times d}$ be the matrix representation of $C$, where weighted points are duplicated appropriately}
\State{For $i\in[\ell]$, let $\bG^{(i)}$ be the group of rows (removing multiplicity) of $\widetilde{\bA}$ with frequency $[(1+\eps)^i,(1+\eps)^{i+1})$}
\State{Let $(\widetilde{\bU^{(i)}},\widetilde{\bV^{(i)}})$ be the output of \lemref{lem:l0:sample:matrix} on input $\bG^{(i)}$, padded to $\mathbb{R}^{n\times k}$ and $\mathbb{R}^{k\times d}$, respectively}
\State{$\widetilde{\bV}\gets\widetilde{\bV^{(0)}}\circ\ldots\circ\widetilde{\bV^{(\ell)}}$}
\State{Use \algref{alg:compute:u:gf2} with $\widetilde{\bV^{(0)}},\ldots,\widetilde{\bV^{(\ell)}}$ and $\bA$ to find $\bU'$}
\State{\Return $(\bU',\bV')$ with $\bV'=\widetilde{\bV}$}
\end{algorithmic}
\end{algorithm}

We first justify the correctness of \algref{alg:lra:gf2}. 
\begin{lemma}
\lemlab{lem:correctness:lra:gf2}
With probability at least $0.95$, \algref{alg:lra:gf2} returns $\bU',\bV'$ such that
\[\|\bU'\bV'-\bA\|_F^2\le(1+\eps)\min_{\bU\in\{0,1\}^{n\times k},\bV\in\{0,1\}^{k\times d}}\|\bU\bV-\bA\|_F^2,\]
where all matrix operations are performed in $\mathbb{F}_2$. 
\end{lemma}
\begin{proof}
Let $\widetilde{\bU}\gets\begin{bmatrix}\widetilde{\bU^{(0)}}|\ldots|\widetilde{\bU^{(\ell)}}\end{bmatrix}$ in \algref{alg:lra:gf2}. 
Let $\widetilde{\bM}$ be the indicator matrix that selects a row of $\widetilde{\bU}\widetilde{\bV}=\widetilde{\bU}\bV'$ to match to each row of $\bA$, so that by the optimality of $\bU'$, 
\[\|\bU'\bV'-\bA\|_F^2\le\|\widetilde{\bM}\widetilde{\bU}\widetilde{\bV}-\bA\|_F^2.\]

Since $\bV$ is a set of $k$ points in $\{0,1\}^d$ and each row $\bU_i$ of $\bU$ induces one of at most $2^k$ possible points $\bU_i\bV\in\{0,1\}^d$, then $\|\bU\bV-\bA\|_F^2$ is the objective value of a constrained $2^k$-means clustering problem, even when all operations performed are on $\mathbb{F}_2$. 
Similarly, $\bV^{(j)}$ is a set of $k$ points in $\{0,1\}^d$ for each $j\in[\ell]$. 
Each row $\bU_i$ of $\bU$ induces one of at most $2^k$ possible points $\bU_i\bV^{(j)}\in\{0,1\}^d$ for a fixed $j\in[\ell]$, so that $\|\bU\bV'-\bA\|_F^2$ is the objective value of a constrained $2^k\ell$-means clustering problem, even when all operations performed are on $\mathbb{F}_2$. 

Hence by the choice of $t$ in \thmref{thm:strong:coreset}, it follows that $\widetilde{\bA}$ is a strong coreset, and so
\[\|\widetilde{\bM}\widetilde{\bU}\widetilde{\bV}-\bA\|_F^2\le(1+\eps)\|\widetilde{\bU}\widetilde{\bV}-\widetilde{\bA}\|_F^2.\] 
We decompose the rows of $\widetilde{\bA}$ into $\bG^{(0)},\ldots,\bG^{(\ell)}$ for $\ell=\O{\frac{\log n}{\eps}}$. 
Let $G_i$ be the corresponding indices in $[n]$ so that $j\in G_i$ if and only if $\widetilde{\bA_j}$ is nonzero in $\bG_i$. 
Then we have
\[\|\widetilde{\bU}\widetilde{\bV}-\widetilde{\bA}\|_F^2=\sum_{i\in[\ell]}\sum_{j\in G_i}\|\bU'_j\bV'-\widetilde{\bA_j}\|_F^2.\]
Since each row in $G_i$ is repeated a number of times in $[(1+\eps)^i,(1+\eps)^{i+1})$, then
\[\sum_{j\in G_i}\|\bU'_j\bV'-\widetilde{\bA_j}\|_F^2\le(1+\eps)^2\min_{\bU^{(i)}\in\{0,1\}^{n\times k},\bV^{(i)}\in\{0,1\}^\times{k\times d}}\|\bU^{(i)}\bV^{(i)}-\bG^{(i)}\|_F^2,\]
where the first factor of $(1+\eps)$ is from the $(1+\eps)$-approximation guarantee of $\bU^{(i)}$ and $\bV^{(i)}$ by \lemref{lem:l0:sample:matrix} and the second factor of $(1+\eps)$ is from the number of each row in $\bG^{(i)}$ varying by at most a $(1+\eps)$ factor. 
Therefore, 
\begin{align*}
\|\bU'\bV'-\bA\|_F^2&\le(1+\eps)^3\sum_{i\in[\ell]}\min_{\bU^{(i)}\in\{0,1\}^{n\times k},\bV^{(i)}\in\{0,1\}^{k\times d}}\|\bU^{(i)}\bV^{(i)}-\bG^{(i)}\|_F^2\\
&\le(1+\eps)^3\min_{\bU\in\{0,1\}^{n\times k},\bV\in\{0,1\}^{k\times d}}\|\bU\bV-\widetilde{\bA}\|_F^2.
\end{align*}
Let $\bU^*\in\{0,1\}^{n\times k}$ and $\bV^*\in\{0,1\}^{k\times d}$ such that 
\[\|\bU^*\bV^*-\bA\|_F^2=\min_{\bU\in\{0,1\}^{n\times k},\bV\in\{0,1\}^{k\times d}}\|\bU\bV-\bA\|_F^2,\]
where all operations are performed in $\mathbb{F}_2$. 
Let $\bM^*$ be the indicator matrix that selects a row of $\bU^*\bV^*$ to match to each row of $\widetilde{\bA}$, so that by \lemref{lem:correctness:subroutine}, 
\[\min_{\bU\in\{0,1\}^{n\times k},\bV\in\{0,1\}^{k\times d}}\|\bU\bV-\widetilde{\bA}\|_F^2\le(1+\eps)\|\bM^*\bU^*\bV^*-\widetilde{\bA}\|_F^2.\]
Then by the choice of $t$ in \thmref{thm:strong:coreset} so that $\widetilde{\bA}$ is a strong coreset of $\bA$,
\[\|\bM^*\bU^*\bV^*-\widetilde{\bA}\|_F^2\le(1+\eps)\|\bU^*\bV^*-\bA\|_F^2.\]
Therefore, we have
\[\|\bU'\bV'-\bA\|_F^2\le(1+\eps)^5\|\bU^*\bV^*-\bA\|_F^2\]
and the desired claim then follows from rescaling $\eps$. 
\end{proof}

It remains to analyze the runtime of \algref{alg:lra:gf2}. 
\begin{lemma}
\lemlab{lem:runtime:lra:gf2}
\algref{alg:lra:gf2} uses $2^{\poly(k/\eps)}\poly(n,d)$ runtime. 
\end{lemma}
\begin{proof}
By \thmref{thm:strong:coreset}, we have that \algref{alg:lra:gf2} uses $\O{nd^2+n^2d+\frac{nkd}{\eps^2}+\frac{nk^2}{\eps^2}}$ time to compute $\widetilde{\bA}\in\{0,1\}^{N\times d}$ with $N=\poly(n)$. 
By \lemref{lem:l0:sample:matrix}, it takes $d\cdot (2^k)^{\poly(k/eps)}$ time to compute $\widetilde{\bU^{(i)}},\widetilde{\bV^{(i)}}$ for each $i\in[\ell]$ for $\ell=\O{\frac{\log n}{\eps}}$. 
Hence, it takes $2^{\poly(k/eps)}\poly(n,d)$ runtime to compute $\widetilde{\bU}$ and $\widetilde{\bV}$. 
Finally, computing $\bU'$ via \algref{alg:compute:u:gf2} takes $\O{2^{k'} n}$ time after enumerating through all possible $2^k\ell$ binary vectors for each row of $\bU'$.
Therefore, the total runtime of \algref{alg:lra} is $2^{\poly(k/\eps)}\poly(n,d)$.
\end{proof}

By \lemref{lem:correctness:lra:gf2} and \lemref{lem:runtime:lra:gf2}, we thus have:
\begin{theorem}
There exists an algorithm that uses $2^{\poly(k/\eps)}\poly(n,d)$ runtime and with probability at least $\frac{2}{3}$, outputs $\bU'\in\{0,1\}^{n\times k'}$ and $\bV'\in\{0,1\}^{k'\times d}$ such that
\[\|\bU'\bV'-\bA\|_F^2\le(1+\eps)\min_{\bU\in\{0,1\}^{n\times k},\bV\in\{0,1\}^{k\times d}}\|\bU\bV-\bA\|_F^2,\]
where $k'=\O{\frac{k\log k}{\eps}}$. 
\end{theorem}

\section{\texorpdfstring{$L_p$}{Lp} Low-Rank Approximation}
In this section, we present a $(1+\eps)$-approximation algorithm for binary low-rank approximation with $L_p$ loss, where to goal is to find matrices $\bU\in\{0,1\}^{n\times k}$ and $\bV\in\{0,1\}^{k\times d}$ to minimize $\|\bU\bV-\bA\|_p^p$. 
We would like to use the same approach as in \secref{sec:binary:lra}, where we first compute a weighted matrix $\widetilde{\bA}$ from a strong coreset for $\bA$, and then we make guesses for the matrices $\bS\bU^*$ and $\bS\bA$ and solve for $\min_{\bV\in\{0,1\}^{k\times d}}\|\bS\bU^*\bV-\bS\bA\|_F^2$ while ensuring there are not too many possibilities for the matrices $\bS\bU^*$ and $\bS\bA$. 
Thus to adapt this approach to $L_p$ loss, we first require the following strong coreset construction for discrete metrics: 

\begin{theorem}[Theorem 1 in \cite{Cohen-AddadSS21}]
\thmlab{thm:strong:coreset:lp}
Let $X\subset\mathbb{R}^d$ be a subset of $n$ points, $\eps\in(0,1)$ be an accuracy parameter, $p\ge 1$ be a constant, and let \[t=\O{\min(\eps^{-2}+\eps^{-p},k\eps^{-2})\cdot k\log n}.\] 
There exists an algorithm that uses $\poly(n,d,k)$ runtime and outputs a set of $t$ weighted points that is a strong $\eps$-coreset for $k$-clustering on discrete $L_p$ metrics with probability at least $0.99$. 
Moreover, each point has an integer weight that is at most $\poly(n)$.
\end{theorem}

For Frobenius error, we crucially require the affine embedding property that 
\[(1-\eps)\|\bU^*\bV-\bA\|_F^2\le\|\bS\bU^*\bV-\bS\bA\|_F^2\le(1+\eps)\|\bU^*\bV-\bA\|_F^2,\]
for all $\bV\in\{0,1\}^{k\times d}$. 
Unfortunately, it is not known whether there exists an efficient sampling-based affine embedding for $L_p$ loss. 

We instead invoke the coreset construction of \thmref{thm:strong:coreset:lp} on the rows and the columns so that $\widetilde{\bA}$ has a small number of distinct rows and columns. 
We again use the idea from \secref{sec:lra:frob:ftwo} to partition the rows of $\widetilde{\bA}$ into groups based on their frequency, but now we further partition the groups based on the frequency of the columns. 
It then remains to solve BMF with $L_p$ loss on the partition, each part of which has a small number of rows and columns. 
Because the contribution of each row toward the overall loss is small (because there is a small number of columns), it turns out that there exists a matrix that samples $\poly(k/\eps)$ rows of each partition that finally achieves the desired affine embedding. 
Thus, we can solve the problem on each partition, pad the factors accordingly, and build the bicriteria factors as in the binary field case. 
The algorithm appears in full in \algref{alg:lra:general}, with subroutines appearing in \algref{alg:compute:u:general} and \algref{alg:distinct:lra:general}.

\begin{algorithm}[!htb]
\caption{Algorithm for computing optimal $\bU$ given $\bV^{(1)},\ldots,\bV^{(\ell)}$}
\alglab{alg:compute:u:general}
\begin{algorithmic}[1]
\Require{$\widetilde{\bA}\in\{0,1\}^{N\times d}$, $\bV^{(1)},\ldots,\bV^{(\ell)}\in\{0,1\}^{k\times d}$}
\Ensure{$\bU'=\argmin_{\bU\in\{0,1\}^{N\times \ell k}}\|\bU\bV-\widetilde{\bA}\|_p^p$, where $\bU$ is restricted to one nonzero block of $k$ coordinates}
\For{$i=1$ to $i=N$}
\State{Set $(\bU'_i,j')=\argmin_{\bU_i\in\{0,1\}^{1\times k},j\in[\ell]}\|(\bU_i\bV^{(j)}-\widetilde{\bA}_i\|_p^p$}
\Comment{Enumerate over all $2^k$ possible binary vectors, all $\ell$ indices}
\State{Pad $\bU'_i$ with length $\ell k$, as the $j'$-th block of $k$ coordinates}
\EndFor
\State{\Return $\bU'=\bU'_1\circ\ldots\circ\bU'_N$}
\end{algorithmic}
\end{algorithm}

\begin{algorithm}[!htb]
\caption{Low-rank approximation for matrix $\widetilde{\bA}$ with $t$ distinct rows and $t'$ distinct columns}
\alglab{alg:distinct:lra:general}
\begin{algorithmic}[1]
\Require{$\widetilde{\bA}\in\{0,1\}^{N\times D}$ with at most $t$ distinct rows and $r$ distinct columns}
\Ensure{$\bU',\bV'$ with $\|\bU\bV-\widetilde{\bA}\|_p\le(1+\eps)\min_{\bU\in\{0,1\}^{N\times k}, \bV\in\{0,1\}^{k\times D}}\|\bU\bV-\widetilde{\bA}\|_p$}
\State{$V\gets\emptyset$}
\For{each guess of $\bS\bU^*$ and $\bS\bA$, where $\bS$ is a $L_0$ sampling matrix with $m=\O{\frac{k^{p+1}}{\eps^2}\log r}$ rows with weights that are powers of two up to $\poly(N)$}
\State{$V\gets V\cup\argmin_{\bV\in\{0,1\}^{k\times D}}\|\bS\bU^*\bV-\bS\bA\|_p^p$}
\Comment{\algref{alg:compute:v} with $L_p$ loss}
\EndFor
\For{each $\bV\in V$}
\State{Let $\bU_\bV=\argmin_{\bU\in\{0,1\}^{N\times k}}\|\bU\bV-\bA\|_p^p$}
\Comment{\algref{alg:compute:u} with $L_p$ loss}
\EndFor
\State{$\bV'\gets\argmin_{\bV\in\{0,1\}^{k\times d}}\|\bS\bU_\bV\bV-\bS\bA\|_p^p$}
\State{$\bU'\gets\bU_{\bV'}$}
\State{\Return $(\bU',\bV')$}
\end{algorithmic}
\end{algorithm}

\begin{algorithm}[!htb]
\caption{Bicriteria low-rank approximation with $L_p$ loss for matrix $\bA$}
\alglab{alg:lra:general}
\begin{algorithmic}[1]
\Require{$\bA\in\{0,1\}^{n\times d}$, rank parameter $k$, accuracy parameter $\eps>0$}
\Ensure{$\bU'\in\{0,1\}^{n\times k},\bV'\in\{0,1\}^{k\times d}$ satisfying the property that $\|\bU'\bV'-\bA\|_p^p\le(1+\eps)\min_{\bU\in\{0,1\}^{n\times k}, \bV\in\{0,1\}^{k\times d}}\|\bU\bV-\bA\|_p^p$}
\State{$t\gets\O{\min(\eps^{-2}+\eps^{-p}, k\eps^{-2}) \cdot k\log n}$} 
\Comment{\thmref{thm:strong:coreset:lp}}
\State{$\ell\gets\O{\frac{\log n}{\eps}}$, $k'\gets\ell k$}
\State{Compute a strong coreset $C$ for $2^k$-means clustering of $\bA$, with $t$ rows, with weights $N=\poly(n)$}
\State{Compute a strong coreset $C'$ for $2^k$-means clustering of $C$, with $t$ rows and columns, with weights $N,D=\poly(n)$}
\State{Let $\widetilde{\bA}\in\{0,1\}^{N\times D}$ be the matrix representation of $C$, where weighted points are duplicated appropriately}
\State{For $i\in[\ell]$, let $\bG^{(i)}$ be the group of rows (removing multiplicity) of $\widetilde{\bA}$ with frequency $[(1+\eps)^i,(1+\eps)^{i+1})$}
\State{For $i,j\in[\ell]$, let $\bG^{(i,j)}$ be the group of columns (removing multiplicity) of $\bG^{(i,j)}$ with frequency $[(1+\eps)^j,(1+\eps)^{j+1})$}
\State{Compute the low-rank minimizers $(\widetilde{\bU^{(i,j)}},\widetilde{\bV^{(i,j)}})$ on input $\bG^{(i,j)}$ using \algref{alg:distinct:lra:general}, padded to $\mathbb{R}^{n\times k}$ and $\mathbb{R}^{k\times D}$, respectively}
\State{$\widetilde{\bU}\gets\begin{bmatrix}\widetilde{\bU^{(0,0)}}|\widetilde{\bU^{(1,0)}}|\ldots|\widetilde{\bU^{(\ell,\ell)}}\end{bmatrix}$, $\widetilde{\bV}\gets\widetilde{\bV^{(0,0)}}\circ\widetilde{\bV^{(1,0)}}\ldots\circ\widetilde{\bV^{(\ell,\ell)}}$}
\State{Use \algref{alg:compute:u:general} with $\widetilde{\bU^{(0,0)}},\widetilde{\bU^{(1,0)}}\ldots,\widetilde{\bU^{(\ell,\ell)}}$ and $C$ to find $\bV'$}
\State{Use $\bV'$ and $\bA$ to find $\bU'$, i.e., \algref{alg:compute:u} with dimension $k'$ and $L_p$ loss}
\State{\Return $(\bU',\bV')$}
\end{algorithmic}
\end{algorithm}

We first justify the correctness of \algref{alg:distinct:lra:general} by showing the existence of an $L_0$ sampling matrix $\bS$ that achieves a subspace embedding for binary inputs. 
\begin{lemma}
\lemlab{lem:lzero:sample}
Given matrices $\bA\in\{0,1\}^{n\times k}$ and $\bB\in\{0,1\}^{n\times r}$, there exists a matrix $\bS\in\mathbb{R}^{m\times n}$ with $m=\O{\frac{k^{p+1}}{\eps^2}\log r}$ such that with probability at least $0.99$, we have that simultaneously for all $\bX\in\{0,1\}^{k\times r}$, 
\[(1-\eps)\|\bA\bX-\bB\|_p^p\le\|\bS\bA\bX-\bS\bB\|_p^p\le(1+\eps)\|\bA\bX-\bB\|_p^p.\]
\end{lemma}
\begin{proof}
Let $\bM\in\{0,1,\ldots,k\}^{n\times 1}$ be an arbitrary matrix and let $S$ be a set that contains the nonzero rows of $\bM$ and has cardinality that is a power of two. 
That is, $|S|=2^i$ for some integer $i\ge 0$. 
Let $\bz$ be a random element of $S$, i.e., a random non-zero row of $\bM$, so that we have
\[\Ex{2^i\cdot\|\bz\|_p^p}=\|\bM\|_p^p.\]
Similarly, we have
\[\Var(2^i\cdot\|\bz\|_p^p)\le 2^ik^p\le 2k^p\|\bM\|_p^p.\]
Hence if we repeat take the mean of $\O{\frac{k^p}{\eps^2}}$ estimators, we have that with probability at least $0.99$, 
\[(1-\eps)\|\bM\|_p^p\le\|\bS\bM\|_p^p\le(1+\eps)\|\bM\|_p^p.\]
We can further improve the probability of success to $1-\delta$ for $\delta\in(0,1)$ by repeating $\O{\log\frac{1}{\delta}}$ times. 
By setting $\bM=\bA\bx-\bB^{(i)}$ for fixed $\bA\in\{0,1\}^{n\times k}$, $\bx\in\{0,1\}^k$, and $\bB\in\{0,1\}^{n\times r}$ with $i\in[r]$, we have that the sketch matrix gives a $(1+\eps)$-approximation to $\|\bA\bx-\bB^{(i)}\|_p^p$. 
The result then follows from setting $\delta=\frac{1}{2^kr}$, taking a union bound over all $\bx\in\{0,1\}^k$, and then a union bound over all $i\in[r]$. 
\end{proof}

We then justify the correctness of \algref{alg:lra:general}. 
\begin{lemma}
\lemlab{lem:correctness:lra:general}
With probability at least $0.95$, \algref{alg:lra:general} returns $\bU',\bV'$ such that
\[\|\bU'\bV'-\bA\|_p^p\le(1+\eps)\min_{\bU\in\{0,1\}^{n\times k},\bV\in\{0,1\}^{k\times d}}\|\bU\bV-\bA\|_p^p.\]
\end{lemma}
\begin{proof}
Let $\bM_1$ and $\bM_2$ be the sampling and rescaling matrices used to acquire $\widetilde{\bA}\in\mathbb{R}^{N\times D}$, so that by the optimality of $\bU'$, 
\[\|\bU'\bV'-\bA\|_p^p\le\|\bM_1\widetilde{\bU}\widetilde{\bV}\bM_2-\bA\|_p^p.\]

Observe that $\bV$ is a set of $k$ points in $\{0,1\}^d$. 
Thus, each row $\bU_i$ of $\bU$ induces one of at most $2^k$ possible points $\bU_i\bV\in\{0,1\}^d$. 
Hence, $\|\bU\bV-\bA\|_p^p$ is the objective value of a constrained $2^k$-clustering problem under the $L_p$ metric. 
Similarly, since $\bV^{(j)}$ is a set of $k$ points in $\{0,1\}^d$ for each $j\in[\ell]$, then each row $\bU_i$ of $\bU$ induces one of at most $2^k$ possible points $\bU_i\bV^{(j)}\in\{0,1\}^d$ for a fixed $j\in[\ell]$. 
Therefore, $\|\bU\bV'-\bA\|_p^p$ is the objective value of a constrained $2^k\ell$-clustering problem under the $L_p$ metric. 

By the choice of $t$ in \thmref{thm:strong:coreset:lp}, $\widetilde{\bA}$ is a strong coreset, and so
\[\|\bM_1\widetilde{\bU}\widetilde{\bV}\bM_2-\bA\|_F^2\le(1+\eps)\|\widetilde{\bU}\widetilde{\bV}-\widetilde{\bA}\|_F^2.\] 
We decompose the rows of $\widetilde{\bA}$ into groups $\bG^{(0)},\ldots,\bG^{(\ell)}$ for $\ell=\O{\frac{\log n}{\eps}}$. 
For each group $\bG^{(i)}$, we decompose the columns of $\bG^{(i)}$ into groups $\bG^{(i,0)},\ldots,\bG^{(i,\ell)}$ for $\ell=\O{\frac{\log n}{\eps}}$. 
Let $G_i$ be the indices in $[n]$ corresponding to the rows in $\bG^{(i)}$ and let $G_{i,j}$ be the indices in $[n]$ corresponding to the columns in $\bG^{(i,j)}$. 
Then
\[\|\widetilde{\bU}\widetilde{\bV}-\widetilde{\bA}\|_p^p=\sum_{i\in[\ell]}\sum_{a\in G_i}\sum_{j\in[\ell]}\sum_{b\in G_{i,j}}\left\lvert(\bU'\bV')_{a,b}-\widetilde{\bA}_{a,b}\right\rvert^p.\]
Since each row in $G_i$ is repeated a number of times in $[(1+\eps)^i,(1+\eps)^{i+1})$ and each column in $G_{i,j}$ is repeated a number of times in $[(1+\eps)^i,(1+\eps)^{i+1})$, then
\[\sum_{a\in G_i}\sum_{b\in G_{i,j}}\left\lvert(\bU'\bV')_{a,b}-\widetilde{\bA}_{a,b}\right\rvert^p\le(1+\eps)^3\min_{\bU\in\{0,1\}^{n\times k},\bV\in\{0,1\}^\times{k\times d}}\sum_{a\in G_i}\sum_{b\in G_{i,j}}\left\lvert(\bU\bV)_{a,b}-\widetilde{\bA}_{a,b}\right\rvert^p,\]
where the first factor of $(1+\eps)$ is from the $(1+\eps)$-approximation guarantee of $\bU^{(i)}$ and $\bV^{(i)}$ by \lemref{lem:l0:sample:matrix} and the second and third factors of $(1+\eps)$ is from the number of each row and each column in $\bG^{(i,j)}$ varying by at most $(1+\eps)$ factor. 
Therefore, 
\begin{align*}
\|\bU'\bV'-\bA\|_p^p&\le(1+\eps)\sum_{i\in[\ell]}\sum_{a\in G_i}\sum_{j\in[\ell]}\sum_{b\in G_{i,j}}\left\lvert(\bU'\bV')_{a,b}-\widetilde{\bA}_{a,b}\right\rvert^p\\
&\le(1+\eps)^4\min_{\bU\in\{0,1\}^{n\times k},\bV\in\{0,1\}^{k\times d}}\|\bU\bV-\widetilde{\bA}\|_p^p\\
\end{align*}
Let $\bU^*\in\{0,1\}^{n\times k}$ and $\bV^*\in\{0,1\}^{k\times d}$ be minimizers to the binary $L_p$ low-rank approximation problem, so that 
\[\|\bU^*\bV^*-\bA\|_p^p=\min_{\bU\in\{0,1\}^{n\times k},\bV\in\{0,1\}^{k\times d}}\|\bU\bV-\bA\|_p^p.\]
Let $\bM_3$ and $\bM_4$ be the indicator matrices that select rows and columns of $\bU^*\bV^*$ to match to each row of $\widetilde{\bA}$, so that by \lemref{lem:correctness:subroutine}, 
\[\min_{\bU\in\{0,1\}^{n\times k},\bV\in\{0,1\}^{k\times d}}\|\bU\bV-\widetilde{\bA}\|_p^p\le(1+\eps)\|\bM_3\bU^*\bV^*\bM_4-\widetilde{\bA}\|_p^p.\]
Then by the choice of $t$ in \thmref{thm:strong:coreset:lp} so that $\widetilde{\bA}$ is a strong coreset of $\bA$,
\[\|\bM_3\bU^*\bV^*\bM_4-\widetilde{\bA}\|_p^p\le(1+\eps)\|\bU^*\bV^*-\bA\|_p^p.\]
Therefore, 
\[\|\bU'\bV'-\bA\|_p^p\le(1+\eps)^6\|\bU^*\bV^*-\bA\|_p^p\]
and the desired claim then follows from rescaling $\eps$. 
\end{proof}

We now analyze the runtime of \algref{alg:lra:general}. 
\begin{lemma}
\lemlab{lem:runtime:lra:general}
For any constant $p\ge 1$, \algref{alg:lra:general} uses $2^{\poly(k/\eps)}\poly(n,d)$ runtime. 
\end{lemma}
\begin{proof}
By \thmref{thm:strong:coreset:lp}, we have that \algref{alg:lra:general} uses $2^{\O{k}}\cdot\poly(n,d)$ time to compute $\widetilde{\bA}\in\{0,1\}^{N\times D}$ with $N,D=\poly(n)$. 
We now consider the time to compute $\widetilde{\bU^{(i,j)}},\widetilde{\bV^{(i,j)}}$ for each $i,j\in[\ell]$ for $\ell=\O{\frac{\log n}{\eps}}$. 
For each $i,j$, we make guesses for $\bS\bU^*$ and $\bS\bA$ in 
Since $\bS\bU^*$ and $\bS\bA$ have $m=\O{\frac{k^{p+1}\log r}{\eps^2}}$ rows, then there are $\binom{t}{m}$ possible choices for $\bS\bU^*$ and $\binom{t}{m}$ choices for $\bS\bA$, where $t=\frac{2^k\log n}{\eps^p}$. 
Hence, there are $2^{\poly(k/\eps)}\poly(n,d)$ possible guesses for $\bS\bU^*$ and $\bS\bA$. 

For each guess of $\bS\bU^*$ and $\bS\bA$, \algref{alg:distinct:lra:general} iterates through the columns of $\widetilde{\bV^{(i,j)}}$, which uses $2^{\O{k}}\cdot\poly(n,d)$ time. 
Similarly, the computation of $\widetilde{\bU^{(i,j)}}$, $\bU'$, and $\bV'$ all take $2^{\O{k}}\cdot\poly(n,d)$ time. 
Therefore, the total runtime of \algref{alg:lra:general} is $2^{\poly(k/\eps)}\poly(n,d)$.
\end{proof}

By \lemref{lem:correctness:lra:general} and \lemref{lem:runtime:lra:general}, we thus have:
\begin{theorem}
For any constant $p\ge 1$, there exists an algorithm that uses $2^{\poly(k/\eps)}\poly(n,d)$ runtime and with probability at least $\frac{2}{3}$, outputs $\bU'\in\{0,1\}^{n\times k'}$ and $\bV'\in\{0,1\}^{k'\times d}$ such that
\[\|\bU'\bV'-\bA\|_p^p\le(1+\eps)\min_{\bU\in\{0,1\}^{n\times k},\bV\in\{0,1\}^{k\times d}}\|\bU\bV-\bA\|_p^p,\]
where $k'=\O{\frac{k\log^2 k}{\eps^2}}$. 
\end{theorem}

We note here that the $\poly(k/\eps)$ term in the exponent hides a $k^p$ factor, as we assume $p$ to be a (small) constant.

\section{Applications to Big Data Models}
\seclab{sec:big:data}
This section describes how we can generalize our techniques to big data models such as streaming or distributed models. 

\paragraph{Algorithmic modularization.} 
To adapt our algorithm to the streaming model or the distributed model, we first present a high-level modularization of our algorithm across all applications, i.e., Frobenius binary low-rank approximation, binary low-rank approximation over $\mathbb{F}_2$, and binary low-rank approximation with $L_p$ loss. 
We are given the input matrix $\bA\in\{0,1\}^{n\times d}$ in each of these settings. 
We first construct a weighted coreset $\widetilde{\bA}$ for $\bA$. 
We then perform a number of operations on $\widetilde{\bA}$ to obtain low-rank factors $\widetilde{\bU}$ and $\widetilde{\bV}$ for $\widetilde{\bA}$. 
Setting $\bV'=\widetilde{\bV}$, our algorithms finally use $\bA$ and $\bV'$ to construct the optimal factor $\bU'$ to match $\bV'$. 

\subsection{Streaming Model}
We can adapt our approach to the streaming model, where either the rows or columns of the input matrix arrive sequentially. 
For brevity, we shall only discuss the setting where the rows of the input matrix arrive sequentially; the setting where the columns of the input matrix arrive sequentially is symmetric. 

\paragraph{Formal streaming model definition.} 
We consider the two-pass row-arrival variant of the streaming model. 
In this setting, the rank parameter $k$ and the accuracy parameter $\eps>0$ are given to the algorithm before the data stream. 
The input matrix $\bA\in\{0,1\}^{n\times d}$ is then defined through the sequence of row-arrivals, $\bA_1,\ldots,\bA_n\in\{0,1\}^d$, so that the $i$-th row that arrives in the data stream is $\bA_i$. 
The algorithm passes over the data twice so that in the first pass, it can store some sketch $S$ that uses space sublinear in the input size, i.e., using $o(nd)$ space. 
After the first pass, the algorithm can perform some post-processing on $S$ and then must output factors $\bU\in\{0,1\}^{n\times k}$ and $\bV\in\{0,1\}^{k\times d}$ after being given another pass over the data, i.e., the rows $\bA_1,\ldots,\bA_n\in\{0,1\}^d$. 

\paragraph{Two-pass streaming algorithm.} 
To adapt our algorithm to the two-pass streaming model, recall the high-level modularization of our algorithm described at the beginning of \secref{sec:big:data}.  
The first step is constructing a coreset $\widetilde{\bA}$ of $\bA$. 
Whereas our previous coreset constructions were offline, we now require a streaming algorithm to produce the coreset $\widetilde{\bA}$. 
To that end, we use the following well-known merge-and-reduce paradigm for converting an offline coreset construction to a coreset construction in the streaming model.

\begin{theorem}
\thmlab{thm:merge:reduce}
Suppose there exists an algorithm that, with probability $1-\frac{1}{\poly(n)}$, produces an offline coreset construction that uses $f(n,\eps)$ space, suppressing dependencies on other input parameters, such as $k$ and $p$.  
Then there exists a one-pass streaming algorithm that, with probability $1-\frac{1}{\poly(n)}$, produces a coreset that uses $f(n,\eps')\cdot\O{\log n}$ space, where $\eps'=\frac{\eps}{\log n}$. 
\end{theorem}

In the first pass of the stream, we can use \thmref{thm:merge:reduce} to construct a strong coreset $C$ of $\bA$ with accuracy $\O{\eps}$. 
However, $C$ will have $2^{\poly(k)}\cdot\poly\left(\frac{1}{\eps},\log n\right)$ rows, and thus, we cannot immediately duplicate the rows of $C$ to form $\widetilde{\bA}$ because we cannot have $\log n$ dependencies in the number of rows of $\widetilde{\bA}$. 

After the first pass of the stream, we further apply the respective offline coreset construction, i.e., \thmref{thm:strong:coreset} or \thmref{thm:strong:coreset:lp} to $C$ to obtain a coreset $C'$ with accuracy $\eps$ and a number of rows independent of $\log n$. 
We then use $C'$ to form $\widetilde{\bA}$ and perform a number of operations on $\widetilde{\bA}$ to obtain low-rank factors $\widetilde{\bU}$ and $\widetilde{\bV}$ for $\widetilde{\bA}$. 
Setting $\bV'=\widetilde{\bV}$, we can finally use the second pass of the data stream over $\bA$, along with $\bV'$, to construct the optimal factor $\bU'$ to match $\bV'$. 
Thus the two-pass streaming algorithm uses $2^{\poly(k)}\cdot d\cdot\poly\left(\frac{1}{\eps},\log n\right)$ total space in the row-arrival model. 
For the column-arrival model, the two-pass streaming algorithm uses $2^{\poly(k)}\cdot n\cdot\poly\left(\frac{1}{\eps},\log d\right)$ total space. 

\subsection{Two-round distributed algorithm.} 
Our approach can also be adapted to the distributed model, where the rows or columns of the input matrix are partitioned across multiple users.  
For brevity, we again discuss the setting where the rows of the input matrix are partitioned; the setting where the columns of the input matrix are partitioned is symmetric. 

\paragraph{Formal distributed model definition.} 
We consider the two-round distributed model, where the rank parameter $k$ and the accuracy parameter $\eps>0$ are known in advance to all users. 
The input matrix $\bA\in\{0,1\}^{n\times d}$ is then defined arbitrarily through the union of rows,  $\bA_1,\ldots,\bA_n\in\{0,1\}^d$, where each row $\bA_i$ may be given to any of $\gamma$ users. 
An additional central coordinator sends and receives messages from the users. 
The protocol is then permitted to use two rounds of communication so that in the first round, the protocol can send $o(nd)$ bits of communication. 
The coordinator can process the communication to form some sketch $S$, perform some post-processing on $S$, and then request additional information from each user, possibly using $o(nd)$ communication to specify the information demanded from each user. 
After the users again use $o(nd)$ bits of communication in the second round of the protocol, the central coordinator must output factors $\bU\in\{0,1\}^{n\times k}$ and $\bV\in\{0,1\}^{k\times d}$. 

\paragraph{Two-round distributed algorithm.} 
To adapt our algorithm to the two-round distributed model, again recall the high-level modularization of our algorithm described at the beginning of \secref{sec:big:data}.  
The first step is constructing a coreset $\widetilde{\bA}$ of $\bA$. 
Whereas our previous coreset constructions were offline, we now require a distributed algorithm to produce the coreset $\widetilde{\bA}$. 
To that end, we request that each of the $t$ users send a coreset with accuracy $\O{\eps}$ of their respective rows. 
Note that each user can construct the coreset locally without requiring any communication since the coreset is only a summary of the rows held by the user. 
Thus the total communication in the first round is just the offline coreset size times the number of players, i.e., $\gamma\cdot 2^{\poly(k)}\cdot\poly\left(\frac{1}{\eps},\log n\right)$ rows. 

Given the union $C$ of the coresets sent by all users, the central coordinator then constructs a coreset $C'$ of $\bA$ with accuracy $\eps$, again using an offline coreset construction. 
The coordinator then uses $C'$ to form $\widetilde{\bA}$ and performs the required operations on $\widetilde{\bA}$ to obtain low-rank factors $\widetilde{\bU}$ and $\widetilde{\bV}$ for $\widetilde{\bA}$. 

The coordinator can then send $\bV'$ to all players, using $\bV'$ and their local subset rows of $\bA$ to construct $\bU'$ collectively. 
The users then send the rows of $\bU'$ corresponding to the rows of $\bA$ local to the user back to the central coordinator, who can then construct $\bU'$. 
Thus the second round of the protocol uses $\tO{nk+kd}\cdot\poly\left(\frac{1}{\eps}\right)$ bits of communication. 
Hence, the total communication of the protocol is $d\gamma\cdot 2^{\poly(k)}\cdot\poly\left(\frac{1}{\eps},\log n\right)+\tO{nk+kd}\cdot\poly\left(\frac{1}{\eps}\right)$ in the two-round row-partitioned distributed model. 
For the two-round column-partitioned distributed model, the total communication of the protocol is $n\gamma\cdot 2^{\poly(k)}\cdot\poly\left(\frac{1}{\eps},\log d\right)+\tO{nk+kd}\cdot\poly\left(\frac{1}{\eps}\right)$. 

\section{Experiments}
In this section, we aim to evaluate the feasibility of the algorithmic ideas of our paper against existing algorithms for binary matrix factorization from previous literature. 
The running time of our full algorithms for BMF is prohibitively expensive, even for small $k$, so our algorithm will be based on the idea of~\cite{KumarPRW19}, who only run their algorithms in part, obtaining weaker theoretical guarantees. Indeed, by simply performing $k$-means clustering, they obtained a simple algorithm that outperformed more sophisticated heuristics in practice.

We perform two main types of experiments, first comparing the algorithm presented in the next section against existing baselines and then showing the feasibility of using coresets in the BMF setting.

\paragraph{Baseline and algorithm.} We compare several algorithms for binary matrix factorization that have implementations available online, namely the algorithm by Zhang et al.~\cite{ZhangLDZ07}, which has been implemented in the \texttt{NIMFA} library~\cite{Zitnik12}, the message passing algorithm of Ravanbakhsh et al.~\cite{RavanbakhshPG16},
as well as our implementation of the algorithm used in the experiments of~\cite{KumarPRW19}. We refer to these algorithms as Zh, MP, and kBMF, respectively.
We choose the default parameters provided by the implementations. We chose the maximum number of rounds for the iterative methods so that the runtime does not exceed 20 seconds, as all methods besides~\cite{KumarPRW19} are iterative. However, in our experiments, the algorithms usually converged to a solution below the maximum number of rounds. We let every algorithm use the matrix operations over the preferred semiring, i.e. boolean, integer, or and-or matrix multiplication, in order to achieve the best approximation.
We additionally found a binary matrix factorization algorithm for sparse matrices based on subgradient descent and random sampling\footnote{\url{https://github.com/david-cortes/binmf}} that is not covered in the literature. 
This algorithm was excluded from our experiments as it did not produce binary factors in our experiments. 
Specifically, we found that it produces real-valued $\bU$ and $\bV$, and requires binarizing the product $\bU\bV$ after multiplication, therefore not guaranteeing that the binary matrix is of rank $k$.

Motivated by the idea of partially executing a more complicated algorithm with strong theoretical guarantees, we build upon the idea of finding a $k$-means clustering solution as a first approximation and mapping the Steiner points to their closest neighbors in $\bA$, giving us a matrix $\bV$ of $k$ binary points, and a matrix $\bU$ of assignments of the points of $\bA$ to their nearest neighbors. 
This solution restricts $\bU$ to have a single non-zero entry per row. 
Instead of outputting this $\bU$ as~\cite{KumarPRW19} did, we solve the minimization problem $\min_{\bU \in \{0, 1\}^{n \times k}} \|\bU \bV - \bA\|_F^2$ exactly at a cost of $2^k$ per row, which is affordable for small $k$. 
For a qualitative example of how this step improves the solution quality, see \figref{congress}. 
We call this algorithm kBMF+.

Using $k$-means as the first step in a binary matrix factorization algorithm is well-motivated by the theoretical and experimental results of~\cite{KumarPRW19}, but does not guarantee a $(1+\eps)$-approximation. 
However, as we do not run the full algorithm, we are not guaranteed a $(1+\eps)$-approximation either way, as unfortunately, guessing the optimal matrix $\bV$ is very time-consuming. 
We would first have to solve the sketched problem $\|\bS\widetilde{\bA} - \bS\bU\bV\|_F^2$ for all guesses of $\bS\bA$ and $\bS\bU$. 

We implement our algorithm and the one of~\cite{KumarPRW19} in Python 3.10 and numpy. 
For solving $k$-means, we use the implementation of Lloyd's algorithm with $k$-means++ seeding provided by the \texttt{scikit-learn} library~\cite{scikit-learn}. All experiments were performed on a Linux notebook with a 3.9 GHz 12th generation Intel Core i7 six-core processor with 32 gigabytes of RAM.

\begin{figure}
    \centering
    \includegraphics[height=9cm]{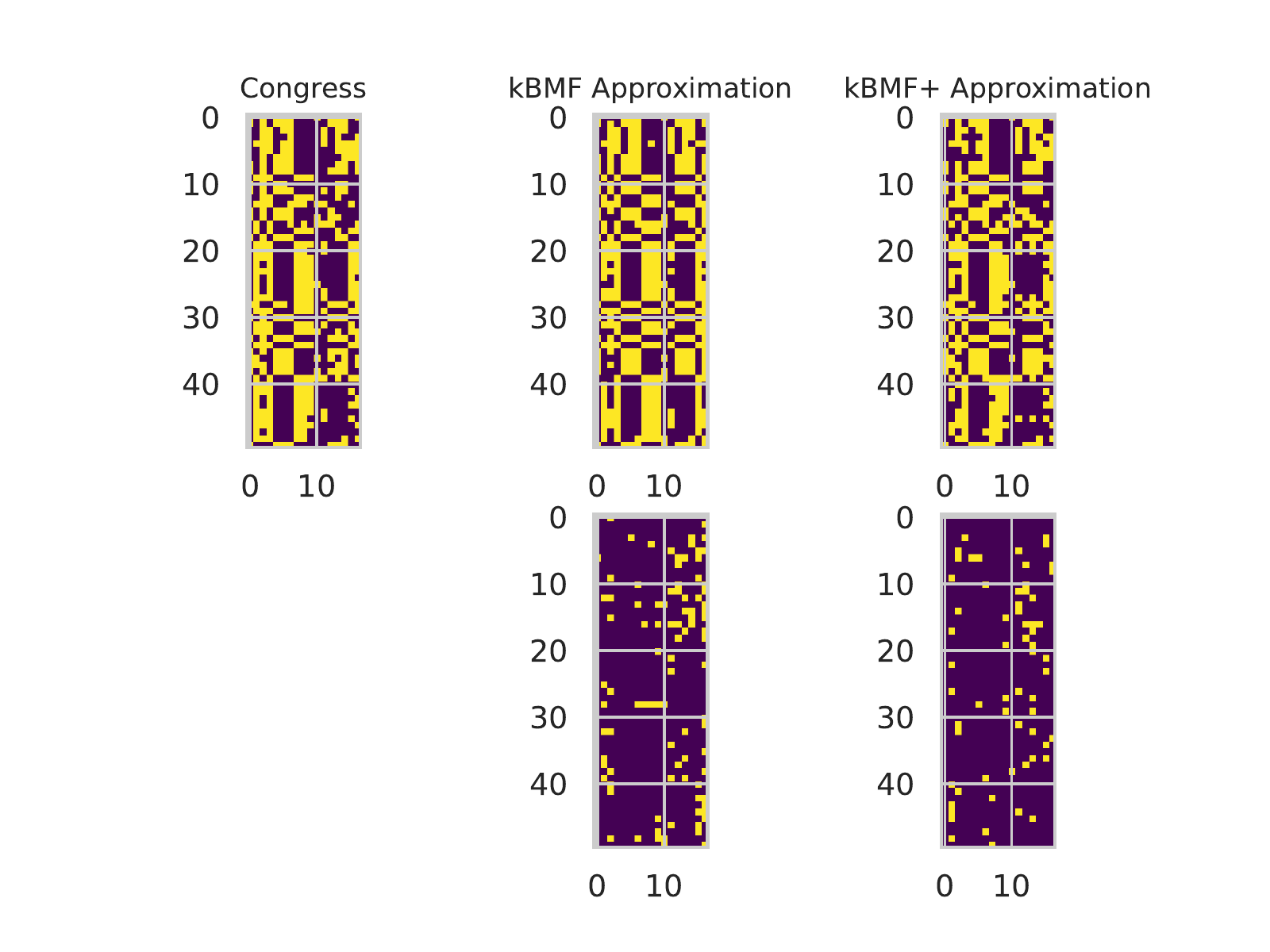}
    \caption{A demonstration of the improved approximation of our algorithm over the algorithm used in the experiments of~\cite{KumarPRW19}. In the first column, we show the first 50 rows of the congress data set, where purple indicates $0$ and yellow indicates $1$.
    The next columns show the approximation of~\cite{KumarPRW19}, and our algorithm's approximation, both with $k=10$. The second row indicates the entries in which the respective approximations differ from the original dataset in yellow. Our experiments found that the number of wrongly reconstructed entries almost halved from the kBMF to the kBMF+ algorithm on this dataset for $k=10$.}
    \figlab{congress}
\end{figure}

\paragraph{Datasets.} We use both real and synthetic data for our experiments. 
We choose two datasets from the UCI Machine Learning Repository~\cite{UCI17}, namely the voting record of the 98th Congress, consisting of 435 rows of 16 binary features representing each congressperson's vote on one of 16 bills, and the Thyroid dataset\footnote{\url{https://www.kaggle.com/datasets/emmanuelfwerr/thyroid-disease-data}}, of 9371 patient data comprising 31 features. 
We restricted ourselves to only binary features, leaving us with 21 columns. 
Finally, we use the ORL dataset of faces, which we binarize using a threshold of $0.33$, as in~\cite{KumarPRW19}.

For our synthetic data, we generate random matrices, where each entry is set to be $1$ independently with probability $p$, at two different sparsity levels of $p \in \{0.1, 0.5\}$. 
Additionally, we generate low-rank matrices by generating $\bU \in \{0, 1\}^{n \times k}$ and $\bV \in \{0,1\}^{k \times d}$ and multiplying them together in $\mathbb{F}_2$. 
We generate $\bU$ and $\bV$ at different sparsity levels of $0.5$ and $0.1$, for $k \in \{5, 10, 15\}$. 
Finally, we also use these matrices with added noise, where after multiplying, each bit is flipped with probability $p_e \in \{0.01, 0.001\}$.

We generate $25$ matrices of size $250 \times 50$ for each configuration. 
These classes are named, in order of introduction: full, lr, and noisy.

\paragraph{Limitations. } We opted to use only binary datasets, thus limiting the available datasets for our experiments. 
Because of this, our largest dataset's size is less than $10000$. 
Our algorithms are practical for these sizes and the parameters $k$ we have chosen. 
Investigating the feasibility of algorithms for binary matrix factorization for large datasets may be an interesting direction for future research.

\subsection{Comparing Algorithms for BMF}

\paragraph{Synthetic data.} For each algorithm, \tableref{synthetic} shows the mean Frobenius norm error (i.e. $\mathrm{err}_{\bA}(\bU, \bV) = \|\bU\bV - \bA\|_F$) across $10$ runs of each algorithm and the mean runtime in milliseconds for the synthetic datasets described above. 
For our choices of parameters, we find that all algorithms terminate in under a second, with Zhang's algorithm and BMF being the fastest and the message-passing algorithm generally being the slowest. 
This is, of course, also influenced by the fact that the algorithms' implementations use different technologies, which limits the conclusions we can draw from the data. 
We find that the kBMF+ algorithm slows down by a factor of $1.5$ for small $k \in \{2, 3, 5\}$, and $15$ when $k = 15$, compared to the kBMF algorithm.

This is offset by the improved error, where our algorithm kBMF+ generally achieves the best approximation for dense matrices, being able to sometimes find a perfect factorization, for example, in the case of a rank 5 matrix, when using $k \in \{10, 15\}$. 
Even when the perfect factorization is not found, we see that the Frobenius norm error is 2-10 times lower. 
On spare matrices, we find that Zhang's and the message-passing algorithms outperform kBMF+, yielding solutions that are about 2 times better in the worst case (matrix of rank 5, with sparsity 0.1 and $k=5$). 
The kBMF algorithm generally performs the worst across datasets, which is surprising considering the results of~\cite{KumarPRW19}. 
Another point of note is that Zhang's algorithm is tuned for sparse matrices, sometimes converging to factors that yield real-valued matrices. 
If so, we attempted to round the matrix as best we could. 


\begin{table*}
\centering
\scalebox{0.6}{
\begin{tabular}{ll|rrrr|rrrr}
\toprule
                & {} & \multicolumn{4}{l}{Error [Frobenius norm]} & \multicolumn{4}{l}{Time [ms]} \\
                & Alg &   kBMF &  kBMF+ & MP & Zh &   kBMF &   kBMF+ &  MP & Zh \\
Dataset & k &       &       &       &       &       &        &        &       \\
\midrule
Random & 2  &  75.8 &  72.3 &  \textbf{71.3} &  \textbf{71.3} &  11.2 &    8.6 &  280.7 &  11.6 \\
$p = 0.5$       & 3  &  74.3 &  69.9 &  69.4 &  \textbf{68.7} &  14.9 &   12.5 &  309.8 &  11.7 \\
                & 5  &  72.2 &  65.8 &  66.6 &  \textbf{64.9} &  10.9 &   11.5 &  347.7 &  13.3 \\
                & 10 &  68.7 &  \textbf{57.4} &  61.5 &  58.5 &  15.4 &   53.4 &  486.6 &  17.2 \\
                & 15 &  66.4 &  \textbf{50.4} &  57.9 &  53.7 &  16.2 &  272.1 &  667.3 &  21.7 \\
\midrule
Random & 2  &  36.0 &  35.0 &  \textbf{34.9} &  35.2 &  10.8 &   11.3 &  277.3 &   9.9 \\
$p=0.1$                & 3  &  35.9 &  \textbf{34.9} &  \textbf{34.9} &  35.0 &   7.5 &   13.9 &  302.1 &  10.6 \\
                & 5  &  35.6 &  34.6 &  35.5 &  \textbf{34.2} &  12.7 &   18.5 &  336.9 &  12.6 \\
                & 10 &  35.0 &  33.9 &  35.8 &  \textbf{31.7} &  17.0 &   64.5 &  459.6 &  15.9 \\
                & 15 &  34.3 &  33.0 &  38.5 &  \textbf{29.0} &  20.9 &  269.5 &  628.4 &  19.6 \\
\midrule
Low-Rank & 2  &  72.5 &  67.1 &  \textbf{66.0} &  67.8 &   4.1 &    7.9 &  274.9 &  11.9 \\
$r=5$           & 3  &  69.2 &  \textbf{60.0} &  62.3 &  64.0 &  12.8 &   12.0 &  301.5 &  13.5 \\
$p=0.5$         & 5  &  64.0 &  \textbf{26.9} &  55.2 &  56.7 &  10.4 &   11.9 &  339.8 &  15.4 \\
                & 10 &  52.9 &   \textbf{0.7} &  41.0 &  42.5 &  14.7 &   72.7 &  472.5 &  19.5 \\
                & 15 &  43.3 &   \textbf{0.0} &  32.8 &  31.1 &  18.0 &  296.0 &  658.0 &  23.8 \\
\midrule
Low-Rank & 2  &  20.5 &  20.4 &  16.5 &  \textbf{15.8} &   9.4 &    6.3 &  185.6 &   4.8 \\
$r=5$            & 3  &  17.0 &  16.6 &  13.1 &  \textbf{12.0} &   5.0 &    5.8 &  209.1 &  12.3 \\
$p=0.1$         & 5  &  11.1 &   8.4 &   \textbf{4.6} &   5.1 &   7.0 &    8.0 &  275.9 &  14.8 \\
                & 10 &   5.1 &   \textbf{0.0} &   0.7 &   2.3 &  19.3 &   75.0 &  460.5 &  18.1 \\
                & 15 &   1.5 &   \textbf{0.0} &   0.4 &   1.4 &  20.2 &  297.0 &  630.9 &  22.1 \\
\midrule
Low-Rank & 2  &  75.8 &  72.2 &  \textbf{71.1} &  71.7 &  13.4 &   15.5 &  281.2 &  11.5 \\
$r=10$                & 3  &  74.3 &  69.6 &  69.1 &  \textbf{69.0} &  15.8 &   20.0 &  308.0 &  11.7 \\
$p=0.5$                & 5  &  72.0 &  \textbf{64.7} &  66.1 &  64.8 &  20.9 &   19.7 &  345.5 &  13.6 \\
                & 10 &  68.2 &  \textbf{28.4} &  60.2 &  57.9 &  16.2 &   51.4 &  477.8 &  17.3 \\
                & 15 &  65.6 &   \textbf{0.8} &  56.0 &  52.9 &  19.3 &  245.2 &  659.6 &  21.3 \\
\midrule
Low-Rank & 2  &  30.8 &  30.5 &  \textbf{27.6} &  28.5 &  10.0 &   14.3 &  213.4 &   5.7 \\
$r=10$                & 3  &  28.5 &  28.1 &  \textbf{25.2} &  25.5 &  11.1 &   13.3 &  248.5 &  11.5 \\
$p=0.5$                & 5  &  24.7 &  23.2 &  20.4 &  \textbf{19.9} &  13.1 &   18.7 &  292.0 &  13.4 \\
                & 10 &  18.3 &  10.2 &   \textbf{7.6} &   8.8 &  16.4 &   76.2 &  434.6 &  16.9 \\
                & 15 &  15.2 &   \textbf{2.5} &   4.7 &   5.4 &  14.8 &  261.3 &  638.8 &  22.1 \\
\midrule
Low-Rank & 2  &  75.7 &  72.3 &  \textbf{71.2} &  71.3 &  14.5 &   18.6 &  277.6 &  11.3 \\
$r=15$                & 3  &  74.2 &  69.9 &  69.3 &  \textbf{68.7} &  12.7 &   11.1 &  306.5 &  11.7 \\
$p=0.5$               & 5  &  72.1 &  65.7 &  66.6 &  \textbf{64.8} &  15.0 &   19.0 &  339.7 &  13.0 \\
                & 10 &  68.6 &  \textbf{56.5} &  61.5 &  58.4 &  18.7 &   51.4 &  478.3 &  17.2 \\
                & 15 &  66.4 &  \textbf{29.2} &  57.7 &  53.6 &  13.0 &  239.9 &  652.8 &  21.1 \\
\midrule
Low-Rank & 2  &  38.7 &  38.2 &  \textbf{35.6} &  36.5 &  12.1 &   10.4 &  242.2 &   9.7 \\
$r=15$                & 3  &  37.1 &  36.2 &  \textbf{33.7} &  34.2 &  10.0 &   13.0 &  274.1 &  12.8 \\
$p=0.1$                & 5  &  33.7 &  32.2 &  29.8 &  \textbf{29.5} &  13.2 &   17.9 &  313.2 &  14.6 \\
                & 10 &  28.1 &  22.3 &  20.3 &  \textbf{19.8} &  20.2 &   56.3 &  457.3 &  17.9 \\
                & 15 &  25.3 &  14.2 &  \textbf{11.6} &  13.4 &  21.2 &  247.9 &  643.8 &  21.2 \\
\midrule
Noisy & 2  &  75.8 &  72.3 &  \textbf{71.2} &  71.6 &  13.9 &   12.8 &  290.4 &  11.3 \\
$r=10$           & 3  &  74.3 &  69.6 &  69.3 &  \textbf{69.0} &  13.8 &   15.6 &  309.3 &  11.6 \\
$p=0.5$         & 5  &  72.1 &  \textbf{64.7} &  66.2 &  65.0 &  17.6 &   23.8 &  345.8 &  13.6 \\
$p_{noise}=0.001$                & 10 &  68.2 &  \textbf{33.8} &  60.3 &  58.1 &  16.8 &   54.0 &  481.1 &  17.6 \\
                & 15 &  65.6 &   \textbf{4.8} &  56.2 &  53.2 &  18.4 &  247.1 &  661.8 &  21.6 \\
\midrule
Noisy & 2  &  32.5 &  32.1 &  \textbf{29.3} &  30.0 &   6.3 &    9.6 &  223.6 &   7.6 \\
$r=10$          & 3  &  30.0 &  29.5 &  \textbf{26.9} &  27.1 &   6.4 &   10.1 &  255.4 &  11.6 \\
$p=0.1$           & 5  &  26.2 &  24.6 &  22.0 &  \textbf{21.3} &   6.6 &    9.7 &  291.9 &  13.5 \\
$p_{noise}=0.001$ & 10 &  19.8 &  12.0 &   \textbf{9.3} &  10.4 &  16.4 &   67.4 &  441.2 &  18.2 \\
                & 15 &  16.7 &   \textbf{4.9} &   6.8 &   7.2 &  13.9 &  255.0 &  641.8 &  22.4 \\
\midrule
Noisy & 2  &  75.8 &  72.1 &  \textbf{71.0} &  71.7 &   9.7 &   11.4 &  276.1 &  11.4 \\
$r=10$          & 3  &  74.3 &  69.5 &  \textbf{69.0} &  69.1 &  12.1 &   13.3 &  302.4 &  12.0 \\
$p=0.5$                & 5  &  72.0 &  \textbf{64.7} &  66.0 &  64.8 &  12.4 &   12.5 &  338.9 &  13.4 \\
$p_{noise}=0.01$                & 10 &  68.3 &  \textbf{38.2} &  60.2 &  57.9 &  15.0 &   50.7 &  475.0 &  17.2 \\
                & 15 &  65.7 &  \textbf{16.7} &  56.1 &  52.8 &  18.0 &  254.0 &  672.9 &  21.3 \\
\midrule
Noisy & 2  &  33.3 &  33.0 &  \textbf{30.3} &  30.9 &   9.9 &   11.5 &  225.3 &   9.2 \\
$r=10$          & 3  &  31.3 &  30.8 &  28.2 &  \textbf{28.0} &  10.8 &   10.5 &  257.5 &  12.5 \\
$p=0.1$                & 5  &  27.8 &  26.2 &  23.6 &  \textbf{23.4} &   9.4 &   18.3 &  292.1 &  14.3 \\
$p_{noise}=0.01$                & 10 &  22.3 &  16.3 &  \textbf{14.0} &  15.1 &  21.0 &   58.5 &  448.5 &  17.4 \\
                & 15 &  19.9 &  12.5 &  12.5 &  \textbf{12.0} &  20.5 &  260.3 &  645.4 &  21.7 \\
\bottomrule
\end{tabular}}
\caption{The average running time and error for different Binary Matrix Factorization algorithms on synthetic datasets. The minimum Frobenius norm error is marked in bold.}\tablelab{synthetic}
\end{table*}


\paragraph{Real data.} As before, \tableref{real} shows the algorithms' average Frobenius norm error and average running time. 
We observe, that all algorithms are fairly close in Frobenius norm error, with the best and worst factorizations' error differing by about up to a factor of 3 across parameters and datasets. 
Zhang's algorithm performs best on the Congress dataset, while the message-passing algorithm performs best on the ORL and Thyroid datasets. 
The kBMF algorithm generally does worst, but the additional processing we do in kBMF+ can improve the solution considerably, putting it on par with the other heuristics. 
On the Congress dataset, kBMF+ is about 1.1-2 times worse than Zhang's, while on the ORL dataset, it is about 10-30\% worse than the message-passing algorithm. 
Finally, the Thyroid dataset's error is about 10-20\% worse than competing heuristics.

We note that on the Thyroid datasets, which has almost 10000 rows, Zhang's algorithm slows considerably, about 10 times slower than kBMF and even slower than kBMF+ for $k=15$. 
This suggests that for large matrices and small to moderate $k$, the kBMF+ algorithm may actually run faster than other heuristics while providing comparable results. 
The message-passing algorithm slows tremendously, being almost three orders of magnitude slower than kBMF, but we believe this could be improved with another implementation.

\begin{table*}
\centering
\scalebox{0.6}{
\begin{tabular}{ll|rrrr|rrrr}
\toprule
        & {} & \multicolumn{4}{l}{Error [Frobenius norm]} & \multicolumn{4}{l}{Time [ms]} \\
        & Alg & kBMF &  kBMF+ & MP & Zh &   kBMF &   kBMF+ &    MP &  Zh \\
Dataset & k &        &       &       &       &       &        &          &        \\
\midrule
Congress & 2  &   40.0 &  38.8 &  38.8 &  \textbf{36.4} &   2.0 &    3.3 &    280.7 &    6.9 \\
        & 3  &   38.4 &  36.6 &  35.9 &  \textbf{32.7} &   2.3 &    4.1 &    311.2 &   13.6 \\
        & 5  &   35.7 &  32.7 &  31.1 &  \textbf{27.7} &   4.6 &    5.2 &    332.9 &   16.2 \\
        & 10 &   32.7 &  23.9 &  22.5 &  \textbf{18.4} &   3.2 &   16.9 &    407.1 &   22.6 \\
        & 15 &   30.9 &  14.8 &  15.5 &   \textbf{9.6} &   7.4 &  246.7 &    480.5 &   27.5 \\
\midrule
ORL & 2  &   39.4 &  37.8 &  35.9 &  \textbf{33.5} &   2.0 &    2.9 &    203.7 &   11.6 \\
        & 3  &   35.7 &  34.6 &  32.2 &  \textbf{29.7} &   2.9 &    4.7 &    241.6 &   13.1 \\
        & 5  &   31.7 &  30.7 &  27.7 &  \textbf{25.6} &   3.8 &    5.8 &    289.4 &   15.4 \\
        & 10 &   26.4 &  25.7 &  21.6 &  \textbf{21.4} &   4.3 &   22.3 &    415.7 &   19.1 \\
        & 15 &   23.4 &  22.8 &  \textbf{17.8} &  19.7 &   6.1 &  318.0 &    575.5 &   22.2 \\
\midrule
Thyroid & 2  &  106.6 &  98.6 &  \textbf{90.5} &  91.6 &  12.6 &   14.2 &   7063.6 &   44.3 \\
        & 3  &   94.5 &  90.5 &  75.5 &  \textbf{73.9} &  14.4 &   18.7 &   7822.0 &   92.9 \\
        & 5  &   82.7 &  80.4 &  78.5 &  \textbf{61.8} &  31.8 &   25.2 &   8860.2 &  132.1 \\
        & 10 &   66.0 &  55.4 &  54.0 &  \textbf{52.9} &  28.9 &   59.6 &  12686.3 &  241.4 \\
        & 15 &   57.6 &  \textbf{38.9} &  39.2 &  46.7 &  26.7 &  313.4 &  16237.7 &  432.7 \\
\bottomrule
\end{tabular}}

\caption{The average running time and error for different Binary Matrix Factorization algorithms on real datasets, minimum frobenius norm error highlighted in bold.}\tablelab{real}
\end{table*}

\paragraph{Discussion.} In our experiments, we found that on dense synthetic data, the algorithm kBMF+ outperforms other algorithms for the BMF problem. 
Additionally, we found that is competitive for sparse synthetic data and real datasets. 
One inherent benefit of the kBMF and kBMF+ algorithms is that they are very easily adapted to different norms and matrix products, as the clustering step, nearest neighbor search, and enumeration steps are all easily adapted to the setting we want. 
A benefit is that the factors are guaranteed to be either 0 or 1, which is not true for Zhang's heuristic, which does not always converge. 
None of the existing heuristics consider minimization of $L_p$ norms, so we omitted experimental data for this setting, but we note here that the results are qualitatively similar, with our algorithm performing best on dense matrices, and the heuristics performing well on sparse data.

\subsection{Using Coresets with our Algorithm}
Motivated by our theoretical use of strong coresets for $k$-means clustering, we perform experiments to evaluate the increase in error using them. 
To this end, we run the BMF+ algorithm on either the entire dataset, a coreset constructed via importance sampling~\cite{bachem2017practical,BravermanFLSZ21}, or a lightweight coreset~\cite{BachemLK18}. 
Both of these algorithms were implemented in Python. 
The datasets in this experiment are a synthetic low-rank dataset with additional noise (size $5000 \times 50$, rank $5$ and $0.0005$ probability of flipping a bit), the congress, and thyroid datasets.

We construct coresets of size $rn$ for each $r \in \{0.001, 0.005, 0.01, 0.02, 0.05, 0.1, 0.2, \dots, 0.9\}$. 
We sample $10$ coresets at every size and use them when finding $\bV$ in our BMF+ algorithm. 
Theory suggests that the quality of the coreset depends only on $k$ and the dimension of the points $d$, which is why in \figref{coreset}, we observe a worse approximation for a given size of coreset for larger $k$. 
We find that the BMF+ algorithm performs just as well on lightweight coresets as the one utilizing the sensitivity sampling framework. This is expected in the binary setting, as the additive error in the weaker guarantee provided by lightweight coresets depends on the dataset's diameter. 
Thus, the faster, lightweight coreset construction appears superior in this setting.

We observe that using coreset increases the Frobenius norm error we observe by about 35\%, but curiously, on the low-rank dataset, the average error decreased after using coresets. 
This may be due to coreset constructions not sampling the noisy outliers that are not in the low-dimensional subspace spanned by the non-noisy low-rank matrix, letting the algorithm better reconstruct the original factors instead.

Our datasets are comparatively small, none exceeding $1000$ points, which is why, in combination with the fact that the coreset constructions are not optimized, we observe no speedup compared to the algorithm without coresets. 
However, even though constructing the coreset takes additional time, the running time between variants remained comparable. 
We expect to observe significant speedups for large datasets using an optimized implementation of the coreset algorithms. 
Using \emph{off the shelf} coresets provides a large advantage to this algorithm's feasibility compared to the iterative methods when handling large datasets.

\begin{figure}
    \centering
    \includegraphics[width=\textwidth]{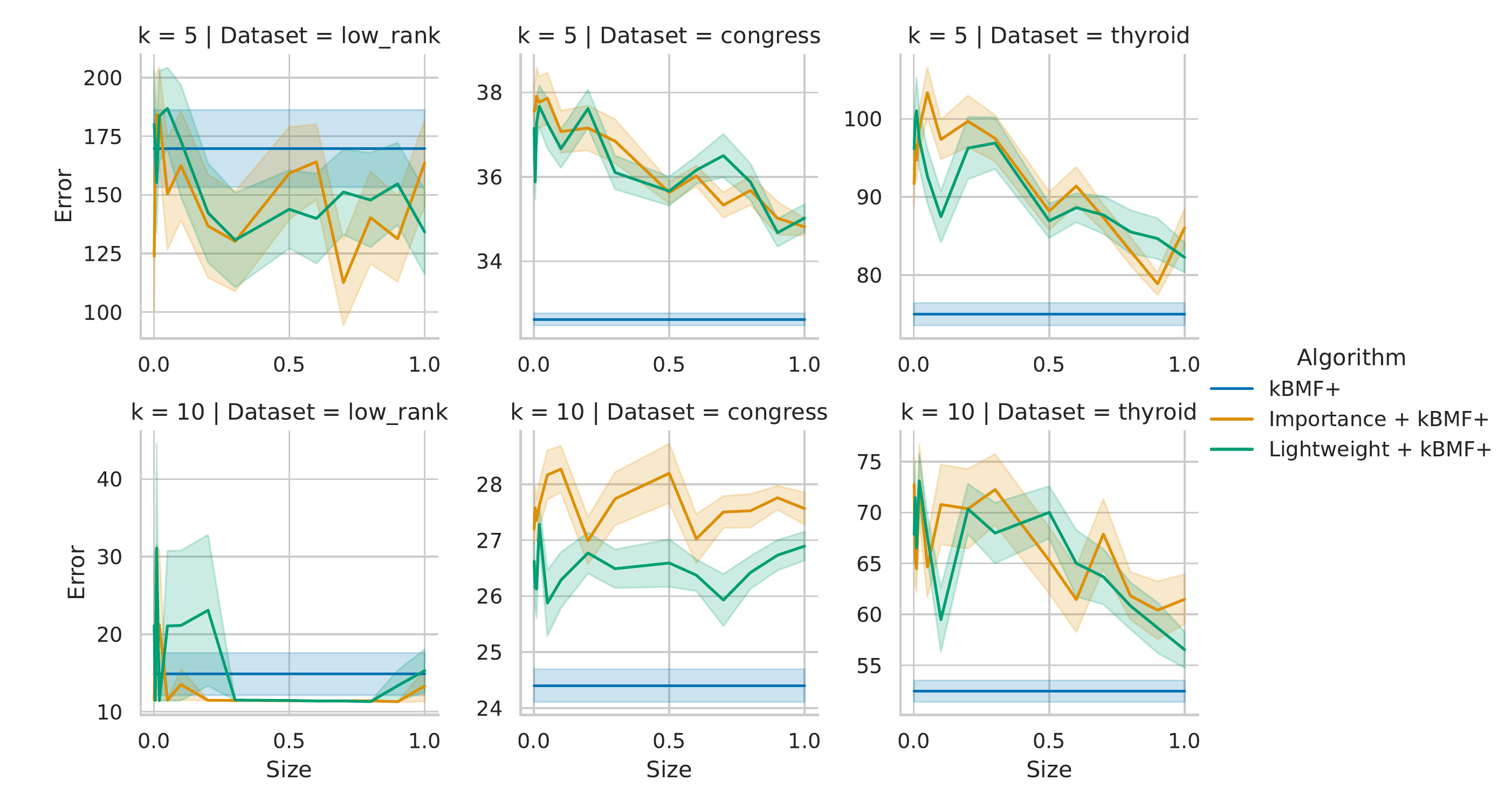}
    \caption{A plot of the effect of different relative coreset sizes on the results of our algorithm.}
    \figlab{coreset}
\end{figure}

\section{Conclusion} 
In this paper, we introduced the first $(1+\eps)$-approximation algorithms for binary matrix factorization with a singly exponential dependence on the low-rank factor $k$, which is often a small parameter. 
We consider optimization with respect to the Frobenius loss, finite fields, and $L_p$ loss. 
Our algorithms extend naturally to big data models and perform well in practice. 
Indeed, we conduct empirical evaluations demonstrating the practical effectiveness of our algorithms. 
For future research, we leave open the question for $(1+\eps)$-approximation algorithms for $L_p$ loss without bicriteria requirements. 

\def\shortbib{0}
\bibliographystyle{alpha}
\bibliography{references}

\end{document}